\documentclass[11pt]{article}
\usepackage{citeref}
\usepackage{amsmath,amssymb,amsfonts,amsthm}
\usepackage{graphicx,subfig,float,epstopdf}
\usepackage{algorithmic}
\usepackage{lineno}
\usepackage{latexsym}

\usepackage{graphicx,subfig,float,epstopdf}
\usepackage{color}
\usepackage[colorlinks, linkcolor=red, anchorcolor=green, citecolor=blue]{hyperref}

\usepackage{slashbox}
\usepackage{booktabs,multirow}
\usepackage{textcomp}
\usepackage{bm}
\usepackage{color}
\allowdisplaybreaks

\numberwithin{equation}{section}

\newtheorem{theorem}{Theorem}
\newtheorem{lemma}{Lemma}

\newtheorem{remark}{Remark}
\newtheorem{definition}{Definition}

\begin{document}

\title{\Large\bf The Dantzig selector: Recovery of  Signal via $\ell_1-\alpha \ell_2$ Minimization
\footnotetext{\hspace{-0.35cm}
\endgraf $^\ast$\,Corresponding author.
\endgraf {1. H. ~Ge is with Sports Engineering College, Beijing Sport University,
Beijing 100084, China (E-mail: gehuanmin@163.com)}
\endgraf {2. P. Li is with  School of Mathematics and Statistics, Gansu Key Laboratory of Applied Mathematics and Complex
Systems, Lanzhou University, Lanzhou 730000, China  (E-mail:lp@lzu.edu.cn)}
}}
\author{Huanmin~Ge$^{1}$ and Peng Li $^{2}$$^{*}$ }
\date{}

\maketitle

\begin{abstract}
In the  paper, we proposed the Dantzig selector based on the $\ell_{1}-\alpha \ell_{2}$ $(0< \alpha \leq1)$ minimization
for the   signal recovery. In the Dantzig selector, the constraint $\|{\bm  A}^{\top}({\bm  b}-{\bm  A}{\bm  x})\|_\infty \leq \eta$ for some small constant $\eta>0$ means the columns of ${\bm  A}$ has very weakly  correlated  with  the error vector ${\bm e}={\bm  A}{\bm  x}-{\bm  b}$. First,   recovery guarantees based on the restricted isometry property (RIP) are established for  signals. Next, we propose the  effective algorithm to solve the proposed  Dantzig selector.
Last, we illustrate the proposed   model and algorithm by extensive numerical experiments  for the recovery of  signals  in the  cases of  Gaussian, impulsive and uniform noise. And the performance of the proposed  Dantzig selector is better than that of the existing methods.

\end{abstract}



\section{Introduction}\label{s1}

\subsection{Signal Recovery}\label{s1.1}
We consider the linear regression model
\begin{equation}\label{systemequationsnoise}
{\bm b}={\bm  A}{\bm  x}+{\bm  e},
\end{equation}
where ${\bm  b}\in\mathbb{R}^m$ are available measurements, the matrix ${\bm  A}\in\mathbb{R}^{m\times n}~(m\ll n)$ models the linear measurement process,
${\bm  x}\in \mathbb{R}^n$ is unknown signal  and ${\bm  e}\in\mathbb{R}^m$ is a vector of measurement errors. To reconstruct  ${\bm  x}$, the most intuitive approach is to find the sparsest signal in the  set of feasible solutions, that is, one solves 
the $\ell_0$ minimization problem: 
\begin{equation}\label{VectorL0}
\min_{{\bm  x}\in\mathbb{R}^n}\|{\bm  x}\|_0~~\text{subject~ to}~~{\bm  b}-{\bm  A}{\bm  x}\in\mathcal{B},
\end{equation}
where $\|{\bm  x}\|_0$ (it usually is called the $\ell_0$ norm of ${\bm  x}$, but is not
a norm) denotes  the number of nonzero coordinates of ${\bm  x}$, and $\mathcal{B}$ is a bounded set determined by the error structure. However, 
this problem \eqref{VectorL0}
is NP-hard and thus computationally infeasible in high dimensional background.

The underdetermined problem (\ref{systemequationsnoise}) puts forward
both theoretical and computational challenges at the interface of statistics and
optimization (see, e.g., \cite{Donoho2005Stable,2006High,2006On}). In the linear regression model,  the so-called Dantzig selector \cite{candes2007dantzig} was proposed to perform variable selection and model fitting. Its formulation model is
\begin{equation}\label{VectorL1-DS}
\min_{{\bm  x}\in\mathbb{R}^n}~\|{\bm  x}\|_{1}~~\text{subject~ to}~ \|{\bm  A}^{\top}({\bm  b}-{\bm  A}{\bm  x})\|_\infty\leq\eta
\end{equation}
where $\eta\geq 0$ is a tuning or penalty parameter.
In \cite{candes2007dantzig}, the performance of  Dantzig selector was analyzed theoretically by deriving
sharp nonasymptotic bounds on the error of estimated coefficients in the $\ell_2$ norm.

In Dantzig selector, the constraint $\|{\bm  A}^{\top}({\bm  b}-{\bm  A}{\bm  x})\|_\infty \leq \eta$
implies  that the correlation between
the residual vector ${\bm  e}={\bm  A}{\bm  x}-{\bm  b}$ and the columns of ${\bm  A}$ is small
for the small penalty parameter $\eta$.
Moreover, the constraint can be viewed  as a
data fitting term and it does not force the residual ${\bm  e}={\bm  A}{\bm  x}-{\bm  b}$ like the $\ell_2$-bounded Gaussian noise.
The Dantzig selector has a wide range of potential applications, especially in statistics.
In Fig \ref{figure.different-noises} and Table $1$,  we present a graphical illustration for  Gaussian, impulsive and uniform noises.
And  we show their distributions and  probability density functions (PDF)  as following.
\begin{enumerate}
\item [(1)]Gaussian Distribution: ${\bm e}\sim \mathcal{N}(0,\sigma^2)\in\mathbb{R}^{m\times 1}$. The noise is usually modeled by the $\ell_2$ norm, i.e.,  $\|{\bm  {Ax}}-{\bm b}\|_{2}$
    with $\bm{b}=\bm{Ax}+\bm{e}$. The probability density function $p$ of  $\bm e$ is
\begin{equation*}
p(\bm e)=
\frac{1}{\sigma{\sqrt{2\pi }}} \exp^{-{\frac{\bm{e}^{2}}{2\sigma^{2}}}}
\end{equation*}
where  $\sigma$  is standard deviation.
\item [(2)] Distribution of impulsive noise: the  distribution  is  Symmetric $\tilde{\alpha}$-stable (S$\tilde{\alpha}$ S ) distribution, which has been  used to model impulsive noise in \cite{wang2020group,wang2013l1,wen2016robust,wen2017efficient}.
The noise is usually modeled by the $\ell_1$ norm, i.e., $\|{\bm  A\bm{x}}-{\bm b}\|_{1}$ with $\bm{b}=\bm{Ax}+\bm{e}$.  Although one cannot analytically present the probability density function for
a general stable distribution, 
its characteristic function of a zero-location
S$\tilde{\alpha}$S distribution can be expressed as
$$
\phi(\omega)=\exp(i\tilde{\delta}\omega-\gamma^{\tilde{\alpha}}|\omega|^{\tilde{\alpha}}),
$$
where $\tilde{\delta}\in(-\infty,\infty)$ is the location parameter, $\gamma\in(0,\infty)$ is the scale parameter, and $\tilde{\alpha}$ is the characteristic exponent measuring the thickness of the distributional tail with $\tilde{\alpha}\in(0, 2]$. If the value of $\tilde{\alpha}$ is smaller, then the tail of the $S\tilde{\alpha} S$ distribution is thicker and consequently the noise is more impulsive.

\item [(3)]Uniform Distribution: ${\bm e}\sim  \mathcal{U}(-\varsigma,\varsigma)\in\mathbb{R}^{m\times 1}$. Its probability density function is
\begin{equation*}
p(\bm e)=
\begin{cases}
\frac{1}{(2\varsigma)^{m}},&\text{if~} -\varsigma\leq e_j\leq \varsigma,\\
0,&\text{otherwise}.
\end{cases}
\end{equation*}
The noise is usually modeled by the $\ell_{\infty}$ norm, i.e., $\|{\bm  {Ax}}-{\bm b}\|_{\infty}$ with $\bm{b}=\bm{Ax}+\bm{e}$.
The $\ell_{\infty}$ minimization problem arises in curve fitting \cite{williams1993least}, optimal control of partial differential equations \cite{clason2011minimal}, image
compression \cite{bredies2012total,wu2000constrained,zhou2012Restoration}. More results about the uniform noise, see \cite{clason2012fitting,wen2018semi,zhang2019fast}.
\end{enumerate}
In Table $1$, we  display the average of $\|{\bm  A}^{T}{\bm  e}\|_\infty$ over 1000 repeated tests, where ${\bm e}$ is a noisy  vector and $\bm A$ is  measurement matrix. Here , let $\bm{A}$ be
Gaussian matrix  or the  oversampled partial DCT matrix.
\begin{enumerate}
\item[(1)]The  random Gaussian matrix $\bm {A}\in\mathbb{R}^{m\times n}$
satisfies $\bm {A}_i\sim \mathcal{N}(0, \frac{1}{m}),~i=1,\ldots,n$. The random Gaussian matrix is  of particular interest in
the practical and theoretical  research. It
has been  a very active area of recent research in  signal processing \cite{das2012snr,suliman2017snr}
and image processing \cite{guerrero2007image,lou2015weighted}. The  random matrix $\bm{A}$
has small coherence and RIP constants  (see Definition \eqref{def:LowerUpperlpRIP}) with high probability \cite{candes2005decoding,chartrand2008restricted}.
The coherence of a matrix ${\bm  A}$ in \cite{donoho2001uncertainty} is the maximum absolute value of
the cross-correlations between the columns of ${\bm  A}$, namely,
$$
\mu({\bm  A}):=\max_{i\neq j}\frac{|\langle {\bm  A}_i,{\bm  A}_j\rangle|}{\|{\bm  A}_i\|_2\|{\bm  A}_j\|_2}.
$$

\item[(2)]The randomly oversampled partial DCT matrix  $\bm {A}\in\mathbb{R}^{m\times n}$ satisfies
$$\bm {A}_i=\frac{\cos(\frac{2\pi\xi}{ F})}{\sqrt{m}},\ \ i=1,\ldots,n,$$
where $\xi\in\mathbb{R}^m\sim\mathcal{U}([0,1]^m)$ which means $\xi $  uniformly and independently   distributes in $[0,1]^m$, and $F\in\mathbb{N}$ is the refinement factor.  Actually,  it is the real part of the random partial Fourier matrix analyzed in \cite{fannjiang2012coherence}. The number $F$ is closely related to the conditioning of $\bm {A}$ in the sense that $\mu(\bm {A})$ tends to get larger as $F$ increases. For example, for $\bm {A}\in\mathbb{R}^{32\times 640}$, $\mu(\bm {A})\approx 0.97$ when $F =5$, and $\mu(\bm {A})$ easily exceeds 0.99 when $F =10$.
The over-sampled DCT matrices  are derived from the problem of
spectral estimation \cite{fannjiang2012coherence,liao2016music} in signal processing, and  radar imagining \cite{fannjiang2012coherence,fannjiang2013compressive} and  surface scattering \cite{fannjiang2012compressive} in  image processing.
\end{enumerate}

\begin{figure}[t] 
\begin{tabular}{ccc}
\includegraphics[width=3.5cm]{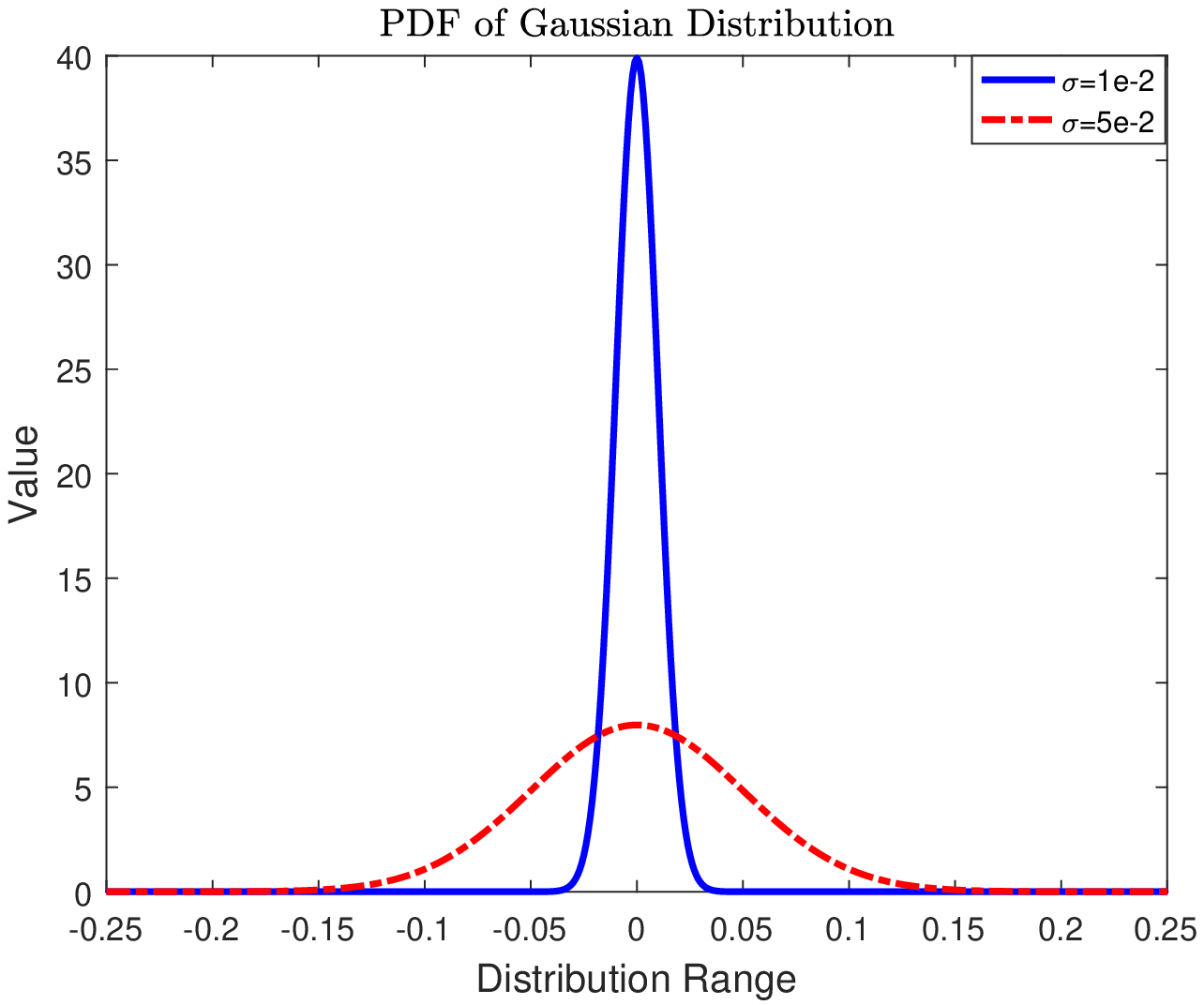}&
\includegraphics[width=3.5cm]{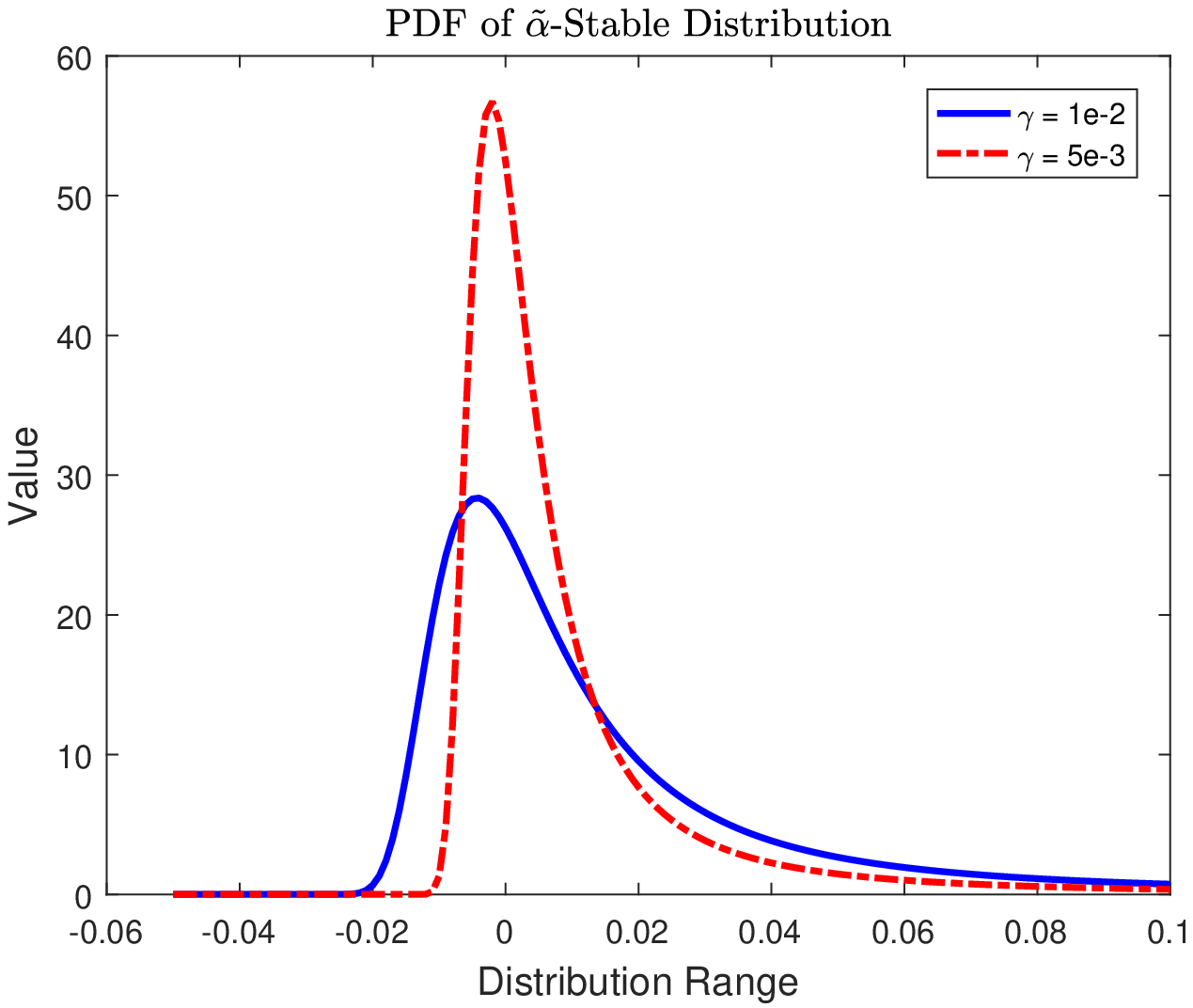}&
\includegraphics[width=3.5cm]{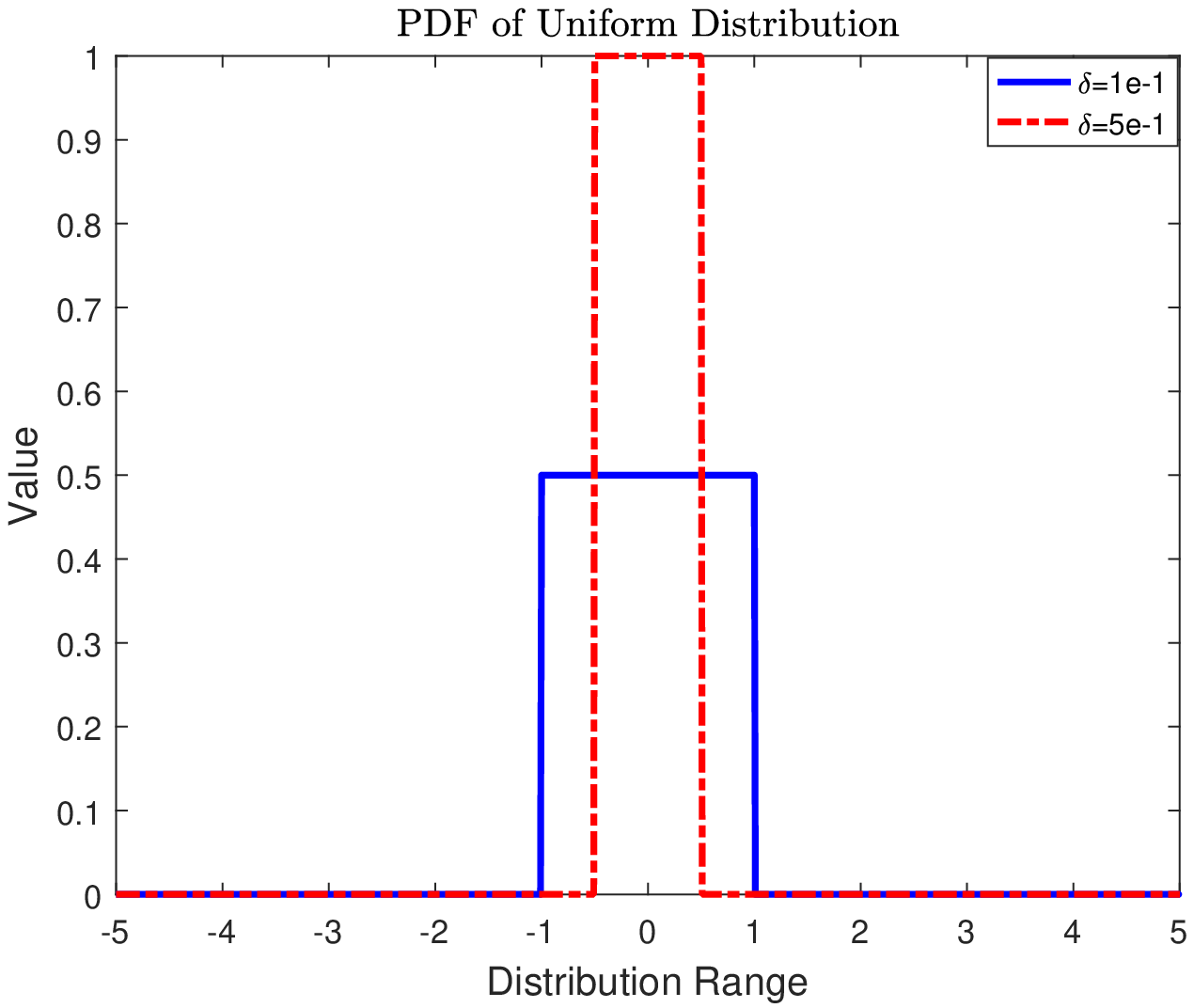}\\
\end{tabular}
\centering
\caption{\rm  Probability density function (PDF) for   Gaussian, impulsive and uniform noises.
Left: Gaussian noise with $\sigma=10^{-2},5\times 10^{-2}$;
 Middle: S$\tilde{\alpha}$S type impulsive noise with $\tilde{\alpha}=1,\tilde{\delta}=0$, $\gamma=10^{-2},5\times 10^{-3}$;
  Right: Uniform noise with $\varsigma=10^{-1},5\times 10^{-1}$.}
\label{figure.different-noises}
\end{figure}

\begin{table}[t] 
\setlength{\tabcolsep}{9pt}
	\caption{\rm The average  of  $\|{\bm  A}^{T}{\bm  e}\|_\infty$ over $10^4$ repeated tests. Let  ${\bm  e}={\bm  A}{\bm  x}-{\bm  b}$, $\bm{A}$ be Gaussian matrix with $n=64,m=256$ and the oversampled DCT matrix  with $n=64,m=256, F=10$. Here we take Gaussian noise with noise levels $\sigma=10^{-2},5\times 10^{-2}$; $S\tilde{\alpha}S$ type implusise noise with $\tilde{\alpha}=1, \tilde{\delta}=0$ and $\gamma=5\times 10^{-3},10^{-2}$;  uniform noise with noise levels $\varsigma=10^{-1},5\times 10^{-1}$.}\label{tab:inftynorm-residual}
	\begin{center}
		\begin{tabular}{c |c | c c  }\hline
Noise Type &Parameter
			       &{\rm Gaussian Matrix}  &{\rm Oversampled DCT}   \\  \hline
{\rm Gaussian}      &$\sigma$=1e-2      &0.0295   &0.0270       \\
                  &$\sigma$=5e-2       &0.1476   &0.1350        \\ \hline
{\rm Impulsive}   &$\gamma$=5e-3   &0.6936   &0.3610        \\
             &$\gamma$=1e-2     &1.2604   &1.397 \\ \hline
{\rm Uniform}      &$\varsigma$=1e-1      &0.1710    &0.1571       \\
          &$\varsigma$=5e-1    &0.8563   &0.7850            \\ \hline
			\specialrule{0.0em}{2.0pt}{2.0pt}
		\end{tabular}
	\end{center}
\end{table}

Table \ref{tab:inftynorm-residual} shows $\|{\bm A}^{T}({\bm  {Ax}}-{\bm b})\|_{\infty}$,  which measures  the correlation between
the noisy vector ${\bm e}={\bm  {Ax}}-{\bm b}$ and the columns of $\bm {A}$, where $\bm{e}$ is  Gaussian, impulsive and uniform noise.
It  works efficiently for the three  type noises (see to Section \ref{s5}), which are different from that of the $\ell_2$ norm $\|{\bm  {Ax}}-{\bm b}\|_{2}$, the $\ell_1$ norm $\|{\bm  {Ax}}-{\bm b}\|_{1}$ and the $\ell_{\infty}$ norm $\|{\bm  {Ax}}-{\bm b}\|_{\infty}$.  These norms only work efficiently for their corresponding noises,.i.e., the $\ell_2$ norm $\|{\bm  {Ax}}-{\bm b}\|_{2}$ is only valid for Gaussian noise,  the $\ell_1$ norm $\|{\bm  {Ax}}-{\bm b}\|_{1}$ is only valid for impulsive noise,
and the $\ell_{\infty}$ norm $\|{\bm  {Ax}}-{\bm b}\|_{\infty}$ is only valid for bounded noise.

 Authors in  \cite{1996Regression} proposed the Least Absolute Shrinkage and Selector Operator (Lasso) as follows
\begin{equation}\label{VectorL1-Lasso}
\min_{{\bm  x}\in\mathbb{R}^n}~\lambda\|{\bm  x}\|_1+\frac{1}{2}\|{\bm  A}{\bm  x}-{\bm  b}\|_2^2,
\end{equation}
where $\lambda> 0$ is a parameter to balance the data fidelity term $\frac{1}{2}\|{\bm  A}{\bm  x}-{\bm  b}\|_2^2$ and the objective function $\|{\bm  x}\|_1$.
In some sense, Lasso estimator and Dantzig selector exhibit similar behavior. Essentially, the Dantzig selector model
\eqref{VectorL1-DS} is a linear program while the Lasso model \eqref{VectorL1-Lasso} is a quadratic program.
  For an extensive study on the relation between the
Dantzig selector and Lasso, we refer to a series of discussion papers which have been
published in ``The Annals of Statistics", e.g., \cite{Bickel2007Discussion,cai2007discussion,Cand2007Rejoinder,efron2007discussion,friedlander2007discussion,meinshausen2007discussion,2008Discussion}.
Readers also can refer to \cite[Chapter 8]{elad2010sparse}. 

Many effective algorithms have been developed to  solve Dantzig selector.
For example,  Cand\`{e}s et.al. in \cite{candes2005magic}  apply   primal-dual algorithm ,
 Wang et.al. in \cite{wang2012linearized} use  Linear ADMM, and  Lu et.al. in \cite{lu2012alternating} solve Dantzig selector  on account of   ADMM.   Chatterjee et.al. in \cite{chatterjee2014generalized} also proposed a Generalized Dantzig Selector (GDS) and solved it by ADMM.

\subsection{Contributions}\label{s1.2}

In this paper, we  introduce the following $\ell_{1}-\alpha\ell_{2}$ minimization problem:
\begin{equation}\label{VectorL1-alphaL2-DS}
\min_{{\bm  x}\in\mathbb{R}^n}~\|{\bm  x}\|_{1}-\alpha\|{\bm  x}\|_{2} ~~\text{subject~ to}~ \|{\bm  A}^{\top}({\bm  b}-{\bm  A}{\bm  x})\|_\infty\leq\eta
\end{equation}
for some constant $\eta\geq0$. Denote \eqref{VectorL1-alphaL2-DS} as $\ell_{1}-\alpha\ell_{2}$-DS.
When $\eta=0$, note that
\eqref{VectorL1-alphaL2-DS} could not  reduce to
\begin{equation}\label{VectorL1-alphaL2-Exact}
\min_{{\bm  x}\in\mathbb{R}^n}~\|{\bm  x}\|_{1}-\alpha\|{\bm  x}\|_{2} ~~\text{subject~ to}~ {\bm  A}{\bm  x}={\bm  b},
\end{equation}
which  means that the $\ell_{1}-\alpha\ell_{2}$-DS  \eqref{VectorL1-alphaL2-DS} is different from
the minimization problems such as the $\ell_0$ minimization \eqref{VectorL0}  with the  constraint term  $\|{\bm  A}{\bm  x}-{\bm  b}\|_p\leq \eta$ for $0<p\leq 2$.
In fact, $\eta=0$ means that  ${\bm A}$ is orthogonal with ${\bm  A}{\bm  x}-{\bm  b}$, i.e.,
 $ \bm A^{\top}(\bm{b}-{\bm  A}{\bm  x})=\bm{0}$.

Similarly, the nonconvex  $\ell_{1}-\alpha \ell_{2}$ ($0<\alpha\leq 1$) minimization method was introduced in \cite{liu2017further,lou2018fast} to   recover  $\bm{x}\in\mathbb{R}^n$
\begin{equation}\label{VectorL1-alphaL2}
\min_{\bm{x}\in\mathbb{R}^n}~\|\bm{x}\|_{1}-\alpha\|\bm{x}\|_{2} \quad \text{subject \ to} \quad \bm{b}-\bm{A}\bm{x}\in\mathcal{B}.
\end{equation}
Clearly,  the method  \eqref{VectorL1-alphaL2} with  $\alpha=1$ reduces to the $\ell_{1-2}$ minimization method  \cite{lou2015computing,yin2015minimization}.
Specifically,
Lou et al. \cite{lou2015computing} and Yin et al. \cite{yin2015minimization} respectively  studied  the $\ell_{1-2}$ minimization  under
$\mathcal{B}=\{\bm{0}\}$ and
 $\mathcal{B}=\{\bm{e}: \,\|\bm{e}\|_2\leq\eta\}$, respectively. They obtained sufficient conditions  based on  RIP
for the recovery of $\bm{x}$
 from \eqref{systemequationsnoise} via  the $\ell_{1-2}$ minimization method.
 To solve \eqref{VectorL1-alphaL2},
they  proposed  the unconstrained problem:
\begin{equation}\label{VectorL1-L2-Lasso}
\min_{\bm{x}\in\mathbb{R}^n}~\lambda(\|\bm{x}\|_{1}-\|\bm{x}\|_{2})+\frac{1}{2}\|\bm{A}\bm{x}-\bm{b}\|_2^2,
\end{equation}
and an effective algorithm based on   the different of convex algorithm (DCA) to solve \eqref{VectorL1-L2-Lasso}.
Several numerical examples in \cite{lou2015computing,yin2015minimization} have demonstrated
that the $\ell_{1}-\ell_{2}$ minimization consistently outperforms the $\ell_1$ minimization and the $\ell_p$ minimization in \cite{lai2013improved}
when the
measurement matrix $\bm{A}$ is highly coherent.
In addition, the metric $\ell_{1}-\ell_{2}$ has shown advantages in various applications such as signal processing \cite{geng2020Unconstrained,li2020minimization,wen2019sparse}, point source super-resolution \cite{lou2016point},
image restoration \cite{ge2021new,li2020minimization, lou2015weighted}, matrix completion \cite{ma2017truncated},  uncertainty quantification \cite{hu2021minimization,yan2017sparse} and  phase retrieval \cite{xia2020sparse,yin2015phaseliftoff}.

In this paper, the main  contributions are as followings:
\begin{enumerate}
\item[(i)] We show sufficient conditions under RIP frame
for the recovery of the signal ${\bm  x}$ from \eqref{systemequationsnoise} via the  $\ell_{1}-\alpha\ell_{2}$-DS \eqref{VectorL1-alphaL2-DS}.
\item[(ii)] We propose  an unconstraint penalty problem
 \eqref{VectorL1-alphaL2-DS-penalty}  to solve the $\ell_{1}-\alpha\ell_{2}$-DS \eqref{VectorL1-alphaL2-DS},
 and develop an effective algorithm based on ADMM to solve the proposed  unconstraint penalty problem.
\item[(iii)]  We present  numerical experiments  for  the recover signal in the cases of  Gaussian, impulsive and uniform noises
 to illustrate the performance of the $\ell_1-\alpha\ell_2$DS. As far as we know, this is the first paper which explore the performances of Dantzig selector under different noises.
\end{enumerate}

\subsection{Organization and Notations}\label{s1.3}

The rest of the paper is organized as follows. We recall some definitions and lemmas in Section
\ref{adds2}.
In Section \ref{s2}, the theoretical results based on $(\ell_2,\ell_1)$-RIP frame
are  showed  for the signal recovery   via the $\ell_1-\alpha\ell_2$ minimization \eqref{VectorL1-alphaL2-DS}.
We show sufficient conditions for the stable recovery of signal under $(\ell_2,\ell_2)$-RIP (i.e., classical RIP) frame via the $\ell_1-\alpha\ell_2$ minimization \eqref{VectorL1-alphaL2-DS} in Section \ref{s3}.
Effective algorithms to solve \eqref{VectorL1-alphaL2-DS}
is developed in Section \ref{s4}. In Section \ref{s5}, numerical results for
sparse signals is given. Section \ref{s6} presents a conclusion.

Throughout the paper, we use the following basic notations. Denote the positive integer set by $\mathbb{Z}_+$. Let $|S|$ be the
number of entries in the set $S$. Let $\lceil t \lceil$ be the nearest integer greater than or equal to t.  For any positive integer $n$, let $[[1,n]]$ be the set $\{1,\ldots,n\}$. For ${\bm  x}\in\mathbb{R}^n$,  denote
${\bm  x}_{\max(s)}$ as the vector ${\bm  x}$ with all but the largest $s$ entries in absolute value set to zero, and ${\bm  x}_{-\max(s)}={\bm  x}-{\bm  x}_{\max(s)}$. Let ${\bm  x}_S$ be the vector equal to ${\bm  x}$ on $S$ and to zero on $S^c$. Let $\|{\bm  x}\|_{\alpha,1-2}$ be $\|{\bm  x}\|_1-\alpha\|{\bm  x}\|_2$. Especially,
when $\alpha=1$, denote $\|{\bm  x}\|_{\alpha,1-2}$ with $\|{\bm  x}\|_{1-2}$.
 And we denote $n\times n$ identity matrix by ${\bm  I}_n$ and zeros matrix by ${\bm  O}$. And we denote the transpose of matrix ${\bm  A}$ by ${\bm  A}^{\top}$. Use the phrase ``$s$-sparse vector" to refer to vectors of sparsity at most $s$. We use boldfaced letter denote matrix or vector.
 The ${\bm A}\succeq {\bm O}$ (resp., ${\bm A}\succ {\bm O}$) represents that
the  matrix ${\bm A}$ is  positive semidefinite (resp. positive definite) and denote the set of all positive semidefinite (resp. positive definite) matrices of size $n$ by ${\bm S}^{n}_+$ (resp. ${\bm S}_{++}^n$). Given ${\bm A}\succeq {\bm O}$ of size $n$, we
define  $\langle {\bm u}, {\bm v}\rangle_{\bm A}:={\bm u}^{T}{\bm A}{\bm v}$ and $\|{\bm u}\|_{{\bm A}}:=\sqrt{\langle {\bm u}, {\bm u}\rangle_{\bm A}}$ for vectors ${\bm u},{\bm v}\in\mathbb{R}^n$. For a positive definite matrix ${\bm A}$,
$\langle \cdot, \cdot\rangle_{\bm A}$ and $\|\cdot\|_{\bm A}$ define an inner product and norm on $\mathbb {R}^n$ respectively, which
become the standard inner product $\langle \cdot,  \cdot\rangle$ and Euclidean norm $\|\cdot\|_2$  respectively
when ${\bm A}$ is the identity matrix ${\bm I}$.

\section{Preliminaries}\label{adds2}

In this section, we recall some significant definitions  and lemmas in order to characterize the recovery guarantees of
the  $\ell_{1}-\alpha\ell_{2}$-DS \eqref{VectorL1-alphaL2-DS} for the signal  $\bm{x}$ recovery.
The following definition of restricted $(\ell_2,\ell_p)$-isometry property is
 introduced in \cite{li2020minimization}.
\begin{definition}\label{def:LowerUpperlpRIP}
For $0<p\leq 1$ or $p=2$, $s\in\mathbb{Z}_+$, we define the restricted $\ell_2/\ell_p$ isometry constant pair $(\delta_s^{lb},\delta_{s}^{ub})$ of order $s$ with respect to the measurement matrix ${\bm  A}\in\mathbb{R}^{m\times n}$ as the smallest numbers $\delta_s^{lb}$ and $\delta_{s}^{ub}$ such that
\begin{equation}\label{Vectorlq-RUB}
(1-\delta_s^{lb})\|{\bm  x}\|_2^p\leq\|{\bm  {Ax}}\|_p^p\leq(1+\delta_s^{ub})\|{\bm  x}\|_2^p,
\end{equation}
holds for all $s$-sparse signals ${\bm  x}$.
We say that ${\bm  A}$ satisfies the $(\ell_2,\ell_p)$-RIP  if $\delta_{s}^{lb}$ and $\delta_{s}^{ub}$ are small for reasonably large $s$.
\end{definition}

\begin{remark}\label{lqRIP-Remark}
When $\delta_s^{lb}=\delta_s^{ub}=\delta_s$ and $p=1$,  the $(\ell_2,\ell_p)$-RIP in Definition \ref{def:LowerUpperlpRIP} reduces to  the definition of
 the $\ell_1$-RIP \cite{chartrand2008restricted}
\begin{align}\label{Lq-RIP}
(1-\delta_s)\|{\bm  x}\|_2\leq\|{\bm  A}{\bm  x}\|_1\leq(1+\delta_s)\|{\bm  x}\|_2.
\end{align}
In addition, Gaussian matrix satisfies  the $(\ell_2,\ell_1)$-RIP  with high probability.
\end{remark}

When $\delta_s^{lb}=\delta_s^{ub}=\delta_s$ and $p=2$,  the $(\ell_2,\ell_p)$-RIP in Definition \ref{def:LowerUpperlpRIP}
is  the classic RIP in \cite{candes2005decoding,candes2006stable}.

\begin{definition}
The matrix ${\bm  A}\in\mathbb{R}^{m\times n}$ satisfies  the  $(\ell_2,\ell_2)$-RIP of order $s$  with  constant
$\delta_{s}\in [0,1)$ if
\begin{equation}\label{def:RIP}
(1-\delta_s)\|{\bm  x}\|_2^2\leq \|{\bm  A}{\bm  x}\|_2^2\leq (1+\delta_s)\|{\bm  x}\|_2^2
\end{equation}
holds for all $s$-sparse vectors ${\bm  x}\in \mathbb{R}^n$,  i.e., $\|{\bm  x}\|_0\leq s$, where $s$ is an integer. The smallest constant $\delta_k$ is called as the
the restricted isometry constant (RIC).
When $s$ is not an integer, we define $\delta_s$
as $\delta_{\lceil s\rceil}$.
\end{definition}
Here, we show a lemma
from the proof of \cite[Theorem 3.3]{yan2017sparse}, which is a modified cone
constraint inequality for $\ell_{1}-\alpha\ell_{2}$.

\begin{lemma}\label{ConeconstraintinequalityforL1-L2}
For any vectors $\bm {x}, \hat{\bm {x}}\in \mathbb{R}^n$, let $\bm {h}=\hat{\bm {x}}-\bm {x}$.
Assume that $\|\hat{\bm {x}}\|_{\alpha,1-2}\leq\|\bm {x}\|_{\alpha,1-2}$. Then
\begin{align}\label{e:h-maxsupperbound1}
&\|{\bm  h}_{-\max(s)}\|_1\leq\|{\bm  h}_{\max(s)}\|_1+2\|{\bm  x}_{-\max(s)}\|_1+\alpha\|{\bm  h}\|_2,\\
&\|{\bm  h}_{-\max(s)}\|_1-\alpha\|{\bm  h}_{-\max(s)}\|_2\leq\|{\bm  h}_{\max(s)}\|_1+2\|{\bm  x}_{-\max(s)}\|_1\nonumber\\
&\ \ \ \ \ \ \ \ \ \ \ \ \ \ \ \ \ \ \ \ \ \ \ \ \ \ \ \ \ \ \ \ \ \ \ \  +\alpha\|{\bm  h}_{\max(s)}\|_2.
\label{e:h-maxsupperbound2}
\end{align}
Especially, when ${\bm  x}$ is $s$-sparse, one has
\begin{align}\label{e:h-maxsupperboundnoiseless1}
\|{\bm  h}_{-\max(s)}\|_1&\leq\|{\bm  h}_{\max(s)}\|_1+\alpha\|{\bm  h}\|_2,\\
\|{\bm  h}_{-\max(s)}\|_1-\alpha\|{\bm  h}_{-\max(s)}\|_2&\leq\|{\bm  h}_{\max(s)}\|_1
+\alpha\|{\bm  h}_{\max(s)}\|_2.
\label{e:h-maxsupperboundnoiseless2}
\end{align}
\end{lemma}

The following lemma is the fundamental properties of  $\|{\bm  x}\|_1-\alpha\|{\bm  x}\|_2$ with
$0\leq\alpha\leq 1$. The item (a) is a  generalization of  \cite[Lemma 2.1 (a)]{yin2015minimization} and item (b) is trival.
It will be frequently used in our proofs.
\begin{lemma}\label{LocalEstimateL1-L2}
For any $\bm {x}\in \mathbb{R}^n$, the following statements hold:
\begin{description}
  \item (a)  For $0 \leq \alpha\leq 1$, let $T=\text{supp}({\bm  x})$ and $\|{\bm  x}\|_0=s$, then
\begin{align}\label{e:l12alphas}
(s-\alpha\sqrt{s})\min_{j\in T}|x_j|\leq\|{\bm  x}\|_1-\alpha\|{\bm  x}\|_2\leq(\sqrt{s}-\alpha)\|{\bm  x}\|_2.
\end{align}
\item (b) Let $S, S_1, S_2\subseteq [n]$ satisfy $S=S_1\cup S_2$ and $S_1\cap S_2=\emptyset$, then
\begin{align}\label{e:uplowerbound}
\|\bm {x}_{S_1}\|_{1}-\alpha\|\bm {x}_{S_1}\|_{2}+\|\bm {x}_{S_2}\|_{1}-\alpha\|\bm {x}_{S_2}\|_{2}\leq\|\bm {x}_{S}\|_{1}-
\alpha\|\bm {x}_{S}\|_{2}.
\end{align}
\end{description}
\end{lemma}

\section{Stable Recovery  Under  the (L2,~L1)-RIP Frame}\label{s2}

In this section, we will give a sufficient condition based on $(\ell_2,\ell_1)$-RIP for the stable recovery of $\ell_{1}-\alpha\ell_{2}$-DS \eqref{VectorL1-alphaL2-DS}.

\subsection{Auxiliary Lemmas  Under (L2,~L1)-RIP Frame}\label{s2.1}
Before showing sufficient conditions  based on  $(\ell_2,\ell_1)$-RIP of the  $\ell_{1}-\alpha\ell_{2}$-DS \eqref{VectorL1-alphaL2-DS}
for the recovery of signals, we first develop an auxiliary lemma.

\begin{lemma}\label{hupperbound}
Assume that $\|\hat{{\bm  x}}\|_{\alpha,1-2}\leq\|{\bm  x}\|_{\alpha,1-2}$.
Let ${\bm  h}=\hat{{\bm  x}}-{\bm  x}$, $T_{01}=T_0\cup T_1$, where $T_0={\mathrm{supp}}({\bm  h}_{\max(s)})$ and
$T_1$ be the index set of the $k\in\mathbb{Z}_+$
largest entries of ${\bm  h}_{-\max(s)}$. Then
\begin{align*}
\|{\bm  h}_{T_{01}^c}\|_{2}
\leq&\frac{1}{2\sqrt{t}}\bigg(\|{\bm  h}_{T_{01}}\|_2+\frac{2\|{\bm  x}_{-\max(s)}\|_1
+\alpha\|{\bm  h}\|_2}{\sqrt{s}}\bigg),
\end{align*}
and
\begin{align*}
\|{\bm  h}\|_2
&\leq\bigg(1+\frac{1}{2\sqrt{t}}\bigg)\|{\bm  h}_{T_{01}}\|_2+\frac{1}{2\sqrt{t}}\frac{2\|{\bm  x}_{-\max(s)}\|_1}{\sqrt{s}}
+\frac{1}{2\sqrt{t}}\frac{\alpha\|{\bm  h}\|_2}{\sqrt{s}},
\end{align*}
where $t=\frac{k}{s}$.
\end{lemma}

\begin{proof}
By the fact that $\|{\bm  h}\|_2=\sqrt{\|{\bm  h}_{T_{01}}\|_2^2+\|{\bm  h}_{T_{01}^c}\|_2^2}$, we need to
 estimate the upper bound of $\|{\bm  h}_{T_{01}^c}\|_{2}$.
 Without loss of generality, we assume that $|h_1|\geq\cdots\geq|h_s|\geq|h_{s+1}|\geq\cdots\geq|h_{s+k}|\geq\cdots\geq|h_n|$ with $k=ts\in \mathbb{Z}_+$.
 Then,
\begin{align}\label{hupperbound.eq1}
\|{\bm  h}_{T_{01}^c}\|_{2}
\leq&\sqrt{\|{\bm  h}_{T_{01}^c}\|_1\|{\bm  h}_{T_{01}^c}\|_{\infty}}
\overset{(a)}{\leq}\sqrt{\Big(\|{\bm  h}_{T_0^c}\|_1-\sum_{j\in T_1}|h_j|\Big)|h_{s+k}|}\nonumber\\
\overset{(b)}{\leq}&\sqrt{\Big(\|{\bm  h}_{T_0^c}\|_1-k|h_{s+k}|\Big)|h_{s+k}|}
=\sqrt{-k\Big(|h_{s+k}|-\frac{\|{\bm  h}_{T_0^c}\|_1}{2k}\Big)^2
+\frac{\|{\bm  h}_{T_0^c}\|_1^2}{4k}}\nonumber\\
\leq&\frac{\|{\bm  h}_{T_0^c}\|_1}{2\sqrt{k}}
\overset{(c)}{\leq}\frac{\|{\bm  h}_{\max(s)}\|_1+2\|{\bm  x}_{-\max(s)}\|_1+\alpha\|{\bm  h}\|_2}{2\sqrt{k}}\nonumber\\
\overset{(d)}{\leq}&\frac{1}{2}\sqrt{\frac{s}{k}}\Big(\|{\bm  h}_{T_{01}}\|_2+\frac{2\|{\bm  x}_{-\max(s)}\|_1
+\alpha\|{\bm  h}\|_2}{\sqrt{s}}\Big)\nonumber\\
\overset{(e)}{=}&\frac{1}{2\sqrt{t}}\bigg(\|{\bm  h}_{T_{01}}\|_2+\frac{2\|{\bm  x}_{-\max(s)}\|_1
+\alpha\|{\bm  h}\|_2}{\sqrt{s}}\bigg),
\end{align}
where (a) and (b) are from $T_{01}=T_0\cup T_1$, $|T_1|=k$ and the assumption $|h_1|\geq\cdots\geq|h_s|\geq|h_{s+1}|\geq\cdots\geq|h_{s+k}|\geq\cdots\geq|h_n|$,
(c) follows from Lemma \ref{ConeconstraintinequalityforL1-L2}, (d) is due to $\|{\bm  h}_{\max(s)}\|_1\leq \sqrt{s}\|{\bm  h}_{\max(s)}\|_2$,
$T_0=\mathrm{supp}({\bm  h}_{\max(s)})$ and $T_{01}=T_0\cup T_1$, and
(e) follows from $k=ts\in \mathbb{Z}_+$.

By  \eqref{hupperbound.eq1}, ones have
\begin{align}\label{hupperbound.eq2}
&\|{\bm  h}\|_2=\sqrt{\|{\bm  h}_{T_{01}}\|_2^2+\|{\bm  h}_{T_{01}^c}\|_2^2}\nonumber\\
&\leq\sqrt{\|{\bm  h}_{T_{01}}\|_2^2+\frac{1}{4t}\bigg(\|{\bm  h}_{T_{01}}\|_2+\frac{2\|{\bm  x}_{-\max(s)}\|_1
+\alpha\|{\bm  h}\|_2}{\sqrt{s}}\bigg)^2}\nonumber\\
&\leq\bigg(1+\frac{1}{2\sqrt{t}}\bigg)\|{\bm  h}_{T_{01}}\|_2+\frac{1}{2\sqrt{t}}\frac{2\|{\bm  x}_{-\max(s)}\|_1}{\sqrt{s}}
+\frac{1}{2\sqrt{t}}\frac{\alpha\|{\bm  h}\|_2}{\sqrt{s}},
\end{align}
where the last inequality is due to the basic inequality $\sqrt{a^2+b^2}\leq a+b$ for $a,b\geq 0$.

\end{proof}

Moreover, we recall a vital lemma, which describes the lower bound of $\|{\bm  A}(\hat{{\bm  x}}-{\bm  x})\|_1$. It plays
an important role in the proof of the main result based on  $(\ell_2,\ell_1)$-RIP frame.

\begin{lemma}(\cite[Lemma 2.6]{li2020minimization})\label{LowerBound}
Assume that $\|\hat{{\bm  x}}\|_{\alpha,1-2}\leq\|{\bm  x}\|_{\alpha,1-2}$.
Let ${\bm  h}=\hat{{\bm  x}}-{\bm  x}$, $T_0={\mathrm{supp}}({\bm  h}_{\max(s)})$,
$T_1$ be the index set of the $k\in\mathbb{Z}_+$
largest entries of ${\bm  h}_{-\max(s)}$  and $T_{01}=T_0\cup T_1$,
the matrix ${\bm  A}$ satisfies the  $(\ell_2,\ell_1)$-RIP condition
of $k+s$ order. Then
\begin{align}\label{e:Ahlowerbound}
\|\bm { Ah}\|_1
\geq&\rho_k\|{\bm  h}_{T_{01}}\|_2
-\frac{2(1+\delta_{k}^{ub})\|{\bm  x}_{-\max(s)}\|_1}{\sqrt{k}-\alpha},
\end{align}
where
\begin{align}\label{e:rho1}
\rho_{k}=1-\delta_{k+s}^{lb}-\frac{(1+\delta_k^{ub})}{a(s,k;\alpha)}
\end{align}
and $a(s,k;\alpha)=\frac{\sqrt{k}-\alpha}{\sqrt{s}+\alpha}$.
\end{lemma}

\subsection{Main Result Based on  (L2,~L1)-RIP Frame}\label{s2.2}
Now, we consider the recovery of signals from
\eqref{systemequationsnoise} with
$\|{\bm  A}^{\top}{\bm  e}\|_\infty\leq \eta$ via  the  $\ell_{1}-\alpha \ell_{2}$-DS (\ref{VectorL1-alphaL2-DS}).
\begin{theorem}\label{StableRecoveryviaVectorL1-alphaL2-DS}
Consider ${\bm  b}={\bm  {Ax}}+{\bm  e}$ with $\|{\bm  A}^{\top}{\bm  e}\|_\infty\leq \eta$. For some $s\in[[1,n]]$ and  $0<\alpha\leq 1$, let
$t>0$ such that $ts\in\mathbb{Z}_+$,  $a(s,ts;\alpha)=\frac{\sqrt{ts}-\alpha}{\sqrt{s}+\alpha}>2$ and $b(t, ts; \alpha)=\frac{8(2\sqrt{ts}-\alpha)}{17\alpha(2\sqrt{t}+1)}>1$
satisfying   $a(s,ts;\alpha)b(t, ts; \alpha)<a(s,ts;\alpha)+b(t, ts; \alpha)$. Let $\hat{\bm{x}}^{DS}$ be the minimizer
of the  $\ell_{1}-\alpha \ell_{2}$-DS (\ref{VectorL1-alphaL2-DS}). If the  measurement matrix ${\bm  A}$ satisfies the $(\ell_2,\ell_1)$-RIP condition with
\begin{eqnarray}\label{RIPCondition1}
&\big(b(t,ts;\alpha)+1\big)\delta_{ts}^{ub}+a(s,ts;\alpha)b(t,ts;\alpha)\delta_{(t+1)s}^{lb}\nonumber\\
&<a(s,ts;\alpha)b(t,ts;\alpha)-b(t,ts;\alpha)-1,
\end{eqnarray}
then
\begin{align*}
\|\hat{{\bm  x}}^{DS}-{\bm  x}\|_2
&\leq\frac{\sqrt{s}\tau}{\sqrt{s}-\alpha\tau}\frac{2\|{\bm  x}_{-\max(s)}\|_1}{\sqrt{s}}\\
&\quad+\frac{2(2\sqrt{t}+1)\big((1+\delta_{ts}^{ub})+a(s, ts;\alpha)\rho_{ts}\big)ms}{\sqrt{t}(\sqrt{s}-\alpha\tau)(1+\delta_{ts}^{ub})\rho_{ts}^2}\eta
\end{align*}
where $\tau=\frac{1}{2\sqrt{t}}\bigg(\frac{17(2\sqrt{t}+1)(1+\delta_{ts}^{ub})}{8a(s,ts;\alpha)\rho_{ts}}+1\bigg)$ and $\rho_{ts}=1-\delta_{(t+1)s}^{lb}-\frac{1+\delta_{ts}^{ub}}{a(s,ts;\alpha)}$.
\end{theorem}

\begin{remark}\label{Simplercondition-StableRecoveryviaVectorL1-alphaL2-DS}
The conditions in Theorem \ref{StableRecoveryviaVectorL1-alphaL2-DS} seem strict. In fact,
these  conditions can be satisfied. For example, for $\alpha=1$, if we take $t=16$, then
$$
a(s,ts;\alpha)=\frac{4\sqrt{s}-1}{\sqrt{s}+1}=:a(s), \ \ \ \ \ \ b(t,ts;\alpha)=\frac{8(8\sqrt{s}-1)}{153}=:b(s).
$$
If we restrict $7\leq s\leq 14$, we can check that $a(s)>2$, $b(s)>1$ and $a(s)b(s)<a(s)+b(s)$. Therefore, $(\ell_2,\ell_1)$-RIP condition (\ref{RIPCondition1}) can be formulated as
$$
\big(b(s)+1\big)\delta_{ts}^{ub}+a(s)b(s)\delta_{(t+1)s}^{lb}<a(s)b(s)-b(s)-1.
$$
And if we take $\delta_s^{lb}=\delta_s^{ub}=\delta_s$ in Remark \ref{lqRIP-Remark}, then condition (\ref{RIPCondition1}) can be simplified as
$$
\delta_{17s}<\frac{192s-305\sqrt{s}-137}{320s+113\sqrt{s}+153}.
$$
\end{remark}

\begin{proof}
Our proof is motivated by the proof of \cite[Lemma 7.9 in Supplement]{cai2015rop}.
Take ${\bm  h}=\hat{{\bm  x}}^{DS}-{\bm  x}$.
Since  $\hat{{\bm  x}}^{DS}$ is the minimizer  of (\ref{VectorL1-alphaL2-DS}),
which implies $\|\hat{{\bm  x}}^{DS}\|_{\alpha,1-2}\leq \|{\bm  x}\|_{\alpha,1-2}$
and $\|{\bm  A}^{\top}({\bm  b}-{\bm  A}\hat{{\bm  x}}^{DS})\|_\infty\leq \eta$.
Then, by \eqref{e:h-maxsupperbound2} in Lemma \ref{ConeconstraintinequalityforL1-L2}, 
we have
\begin{equation}\label{Coneconstraintinequality.eq}
\|{\bm  h}_{-\max(s)}\|_1\leq\|{\bm  h}_{\max(s)}\|_1+2\|{\bm  x}_{-\max(s)}\|_1+\alpha\|{\bm  h}\|_2.
\end{equation}

From the facts $\|\bm {A}^{\top}\bm {z}\|_\infty=\|{\bm  A}^{\top}({\bm  b}-{\bm  {Ax}})\|_\infty\leq \eta$
and $\|{\bm  A}^{\top}(\bm {b}-\bm {A}\hat{{\bm  x}}^{DS})\|_\infty\leq \eta$,
we have the following tube constraint inequality
\begin{align}\label{Tubeconstraintinequality}
\|{\bm  A}^{\top}\bm {Ah}\|_{\infty}
&=\|{\bm  A}^{\top}({\bm  A}{\hat{\bm  x}}^{DS}-\bm { Ax)}\|_{\infty}\nonumber\\
&\leq\|{\bm  A}^{\top}({\bm  A}\hat{{\bm  x}}^{DS}-\bm { b})\|_{\infty}+\|{\bm  A}^{\top}({\bm  b}-\bm { Ax})\|_{\infty}\nonumber\\
&\leq\eta+\eta=2\eta.
\end{align}

Let $T_0=\text{supp}({\bm  h}_{\max(s)})$. First, we partition $T_0^c=[[1,n]]\backslash T_0$  as
$$
T_0^c=\bigcup_{j=1}^J T_j,
$$
where $T_1$ is the index set of the $ts\in\mathbb{Z}_+$ largest entries of ${\bm  h}_{-\max(s)}$, $T_2$ is the index set of the next $ts\in\mathbb{Z}_+$ largest entries of ${\bm  h}_{-\max(s)}$, and so on.
Notice that the last index set $T_J$ may contain less $ts\in\mathbb{Z}_+$ elements.
Similarly,
let  $T_{01}=T_0\cup T_1$.
Thus, by  ${\bm  A}$ satisfies  the $(\ell_2,\ell_1)$-RIP condition of $(t+1)s$ order,
 and Lemma \ref{LowerBound} with $k=ts$, 
one obtains a lower bound of $\|\bm {Ah}\|_1$
\begin{eqnarray}\label{LowerBound.eq1}
\|\bm {Ah}\|_1
\geq\rho_{ts}\|\bm {h}_{T_{01}}\|_{2}
-(1+\delta_{ts}^{ub})\frac{2\|\bm {x}_{-\max(s)}\|_1}{\sqrt{ts}-\alpha},
\end{eqnarray}
where
$$
\rho_{ts}=1-\delta_{(t+1)s}^{lb}-\frac{(1+\delta_{ts}^{ub})}{a(s,ts;\alpha)}
$$
with $a(s,ts;\alpha)=\frac{\sqrt{ts}-\alpha}{\sqrt{s}+\alpha}>1$. Furthermore,
\begin{eqnarray*}
&1-\delta_{(t+1)s}^{lb}-\frac{1+\delta_{ts}^{ub}}{a(s,ts;\alpha)}
>1-\delta_{(t+1)s}^{lb}-\frac{(1+b(s,ts;\alpha))(1+\delta_{ts}^{ub})}{a(s,ts;\alpha)b(s,ts;\alpha)}>0
\end{eqnarray*}
where the first and second inequalities are  from  $b(s,ts;\alpha)=\frac{8(2\sqrt{ts}-\alpha)}{17\alpha(2\sqrt{t}+1)}>0$ with $0<\alpha\leq 1$ and  $\eqref{RIPCondition1}$, respectively.

Next, we  estimate the upper bound of $\|\bm {Ah}\|_1$.
By Cauchy-Schwartz inequality,  we have
\begin{align}\label{UpperBound.eq1}
\|\bm {Ah}\|_1
&\leq \sqrt{m}\|\bm {Ah}\|_2
=\sqrt{m}\langle \bm {Ah}, \bm {Ah}\rangle^{1/2}\nonumber\\
&=\sqrt{m}\langle \bm {A}^{\top}\bm {Ah}, \bm {h}\rangle^{1/2}
\leq\sqrt{m}\sqrt{\|\bm {A}^{\top}\bm {Ah}\|_{\infty}\|\bm {h}\|_1}\nonumber\\
&=\sqrt{m}\sqrt{\|\bm {A}^{\top}\bm {Ah}\|_{\infty}(\|\bm {h}_{T_{0}}\|_1+\|\bm {h}_{T_{0}^c}\|_1)}\nonumber\\
&\overset{(a)}{\leq}\sqrt{m}\sqrt{2\eta(2\|\bm {h}_{T_{0}}\|_1+2\|\bm {x}_{-\max(s)}\|_1+\alpha\|\bm {h}\|_2)}\nonumber\\
&\overset{(b)}{\leq}\sqrt{2m\sqrt{s}\eta\bigg(2\|\bm {h}_{T_{01}}\|_2+\frac{2\|\bm {x}_{-\max(s)}\|_1
+\alpha\|\bm {h}\|_2}{\sqrt{s}}\bigg)}
\end{align}
where (a) is from $T_0=\text{supp}({\bm  h}_{\max(s)})$, \eqref{Coneconstraintinequality.eq} and \eqref{Tubeconstraintinequality}, (b) is due to $T_{01}=T_0\cup T_1$ and
 $\|{\bm  h}_{T_{0}}\|_1\leq \sqrt{s}\|{\bm  h}_{T_{0}}\|_2$ with $|T_{0}|\leq s$.

Combining  \eqref{LowerBound.eq1} with  \eqref{UpperBound.eq1}, we have
\begin{eqnarray}\label{e:inequalitynosimiple}
&&\rho_{ts}\|{\bm  h}_{T_{01}}\|_{2}
 -\frac{(1+\delta_{ts}^{ub})\sqrt{s}}{\sqrt{ts}-\alpha}\frac{2\|{\bm  x}_{-\max(s)}\|_1}{\sqrt{s}}\nonumber\\
&&\leq
\sqrt{2m\sqrt{s}\eta\bigg(2\|{\bm  h}_{T_{01}}\|_2+\frac{2\|{\bm  x}_{-\max(s)}\|_1+\alpha\|{\bm  h}\|_2}{\sqrt{s}}\bigg)}.
\end{eqnarray}

To estimate $\|{\bm  h}_{T_{01}}\|_{2}$ from \eqref{e:inequalitynosimiple},
 we consider the following two cases.

\textbf{Case I:}
$$
\rho_{ts}\|{\bm  h}_{T_{01}}\|_{2}
 -\frac{(1+\delta_{ts}^{ub})\sqrt{s}}{\sqrt{ts}-\alpha}\frac{2\|{\bm  x}_{-\max(s)}\|_1}{\sqrt{s}}<0,$$
 i.e.,
\begin{equation}\label{CaseI.Estimation}
\|{\bm  h}_{T_{01}}\|_{2}<
\frac{(1+\delta_{ts}^{ub})\sqrt{s}}{(\sqrt{ts}-\alpha)\rho_{ts}}\frac{2\|{\bm  x}_{-\max(s)}\|_1}{\sqrt{s}}.
\end{equation}

\textbf{Case II:}
\begin{eqnarray*}
\rho_{ts}\|{\bm  h}_{T_{01}}\|_{2}
-\frac{(1+\delta_{ts}^{ub})\sqrt{s}}{\sqrt{ts}-\alpha}\frac{2\|{\bm  x}_{-\max(s)}\|_1}{\sqrt{s}}
\geq0,
\end{eqnarray*}
 which implies
 $$ \|{\bm  h}_{T_{01}}\|_{2}\geq\frac{(1+\delta_{ts}^{ub})\sqrt{s}}{(\sqrt{ts}-\alpha)\rho_{ts}}
 \frac{2\|{\bm  x}_{-\max(s)}\|_1}{\sqrt{s}},
 $$
then the inequality \eqref{e:inequalitynosimiple} is equivalent  to
\begin{eqnarray}\label{e3.13}
&&\bigg(\rho_{ts}\|{\bm  h}_{T_{01}}\|_{2}-\frac{(1+\delta_{ts}^{ub})\sqrt{s}}{\sqrt{ts}-\alpha}\frac{2\|{\bm  x}_{-\max(s)}\|_1}{\sqrt{s}}\bigg)^2\nonumber\\
&&\leq2m\sqrt{s}\eta\bigg(2\|{\bm  h}_{T_{01}}\|_2+\frac{2\|{\bm  x}_{-\max(s)}\|_1+\alpha\|{\bm  h}\|_2}{\sqrt{s}}\bigg).
\end{eqnarray}

Let $X=\|{\bm  h}_{T_{01}}\|_{2}$ and $Y=\frac{2\|{\bm  x}_{-\max(s)}\|_1+\alpha\|{\bm  h}\|_2}{\sqrt{s}}$. By
$\frac{2\|{\bm  x}_{-\max(s)}\|_1}{\sqrt{s}}\leq Y$,
to guarantee  that \eqref{e3.13} holds, it suffices to show
\begin{eqnarray}\label{e3.14}
&\rho_{ts}^2X^2-\bigg(\frac{2\rho_{ts}(1+\delta_{ts}^{ub})\sqrt{s}}{\sqrt{ts}-\alpha}Y+4m\sqrt{s}\eta\bigg)X
-2m\sqrt{s}\eta Y\leq0.
\end{eqnarray}
For the one-variable quadratic inequality $aZ^2-bZ-c\leq 0$ with the  constants  $a,b,c>0$ and $Z\geq0$,
 there is the fact that
$$
Z\leq\frac{b+\sqrt{b^2+4ac}}{2a}\leq\frac{b}{a}+\sqrt{\frac{c}{a}}.
$$
Hence,
\begin{align}\label{e3.15}
X&\leq \frac{2\rho_{ts}\frac{(1+\delta_{ts}^{ub})\sqrt{s}}{\sqrt{ts}-\alpha}Y+4m\sqrt{s}\eta}{\rho_{s,t}^2}
+\sqrt{\frac{2m\sqrt{s}\eta \varepsilon Y}{\rho_{ts}^2\varepsilon }}\nonumber\\
&\leq\frac{2(1+\delta_{ts}^{ub})\sqrt{s}}{(\sqrt{ts}-\alpha)\rho_{ts}}Y+\frac{4m\sqrt{s}}{\rho_{ts}^2}\eta
+\frac{1}{2}\bigg(\frac{2m\sqrt{s}}{\rho_{ts}^2\varepsilon }\eta+\varepsilon Y\bigg)\nonumber\\
&=\bigg(\frac{2(1+\delta_{ts}^{ub})\sqrt{s}}{(\sqrt{ts}-\alpha)\rho_{ts}}+\frac{\varepsilon}{2}\bigg)Y
+\bigg(4+\frac{1}{\varepsilon}\bigg)\frac{m\sqrt{s}}{\rho_{ts}^2}\eta,
\end{align}
where $\varepsilon > 0$ is to be determined later. Here the second inequality comes from the basis inequality $\sqrt{|a||b|}\leq (|a|+|b|)/2$.  Therefore
\begin{align}\label{e:hmaxupperbound1}
\|{\bm  h}_{T_{01}}\|_{2}
\leq&\bigg(\frac{2(1+\delta_{ts}^{ub})\sqrt{s}}{(\sqrt{ts}-\alpha)\rho_{ts}}+\frac{\varepsilon}{2}\bigg)
\frac{2\|{\bm  x}_{-\max(s)}\|_1+\alpha\|{\bm  h}\|_2}{\sqrt{s}}
+\bigg(4+\frac{1}{\varepsilon}\bigg)\frac{m\sqrt{s}}{\rho_{ts}^2}\eta.
\end{align}

Note that
 \begin{eqnarray*}
 &&\bigg(\frac{2(1+\delta_{ts}^{ub})\sqrt{s}}{(\sqrt{ts}-\alpha)\rho_{ts}}+\frac{\varepsilon}{2}\bigg)Y
+\bigg(4+\frac{1}{\varepsilon}\bigg)\frac{m\sqrt{s}}{\rho_{ts}^2}\eta\\
&&=\bigg(\frac{2(1+\delta_{ts}^{ub})\sqrt{s}}{(\sqrt{ts}-\alpha)\rho_{ts}}+\frac{\varepsilon}{2}\bigg)
\frac{2\|{\bm  x}_{-\max(s)}\|_1+\alpha\|{\bm  h}\|_2}{\sqrt{s}}
 +\bigg(4+\frac{1}{\varepsilon}\bigg)\frac{m\sqrt{s}}{\rho_{ts}^2}\eta\\
&&\geq
 \frac{(1+\delta_{ts}^{ub})\sqrt{s}}{(\sqrt{ts}-\alpha)\rho_{ts}}
\frac{2\|{\bm  x}_{-\max(s)}\|_1}{\sqrt{s}},
\end{eqnarray*}
Therefore combining the estimation \eqref{CaseI.Estimation} in \textbf{Case I} and the estimation \eqref{e:hmaxupperbound1} in \textbf{Case II}, one has \eqref{e:hmaxupperbound1} holds for both cases.

By  Lemma \ref{hupperbound}, ones have
\begin{eqnarray}\label{e:hupperbound}
\|{\bm  h}\|_2
\leq\bigg(1+\frac{1}{2\sqrt{t}}\bigg)\|{\bm  h}_{T_{01}}\|_2+\frac{1}{2\sqrt{t}}\frac{2\|{\bm  x}_{-\max(s)}\|_1}{\sqrt{s}}
+\frac{\alpha}{2\sqrt{t}}\frac{\|{\bm  h}\|_2}{\sqrt{s}}.
\end{eqnarray}
Substituting \eqref{e:hmaxupperbound1} into \eqref{e:hupperbound}, ones obtain
\begin{align}\label{ge:h}
\|{\bm  h}\|_2
&\leq\bigg(1+\frac{1}{2\sqrt{t}}\bigg)
\Bigg(
\bigg(4+\frac{1}{\varepsilon}\bigg)\frac{m\sqrt{s}}{\rho_{ts}^2}\eta
+\bigg(\frac{2(1+\delta_{ts}^{ub})\sqrt{s}}{(\sqrt{ts}-\alpha)\rho_{ts}}+\frac{\varepsilon}{2}\bigg)
\frac{2\|{\bm  x}_{-\max(s)}\|_1+\alpha\|{\bm  h}\|_2}{\sqrt{s}}\Bigg)\nonumber\\
&\quad+\frac{1}{2\sqrt{t}}\frac{2\|{\bm  x}_{-\max(s)}\|_1}{\sqrt{s}}
+\frac{\alpha}{2\sqrt{t}}\frac{\|{\bm  h}\|_2}{\sqrt{s}}\nonumber\\
&\leq\frac{1}{2\sqrt{t}}\Bigg((2\sqrt{t}+1)\bigg(\frac{2(\sqrt{s}+\alpha)(1+\delta_{ts}^{ub})}
{(\sqrt{ts}-\alpha)\rho_{ts}}
+\frac{\varepsilon}{2}\bigg)+1\Bigg)\frac{2\|{\bm  x}_{-\max(s)}\|_1}{\sqrt{s}}\nonumber\\
&\quad+\frac{\alpha}{2\sqrt{t}}
\Bigg((2\sqrt{t}+1)\bigg(\frac{2(\sqrt{s}+\alpha)(1+\delta_{ts}^{ub})}{(\sqrt{ts}-\alpha)\rho_{ts}}
+\frac{\varepsilon}{2}\bigg)+1\Bigg)
\frac{\|{\bm  h}\|_2}{\sqrt{s}}\nonumber\\
&\quad+\bigg(1+\frac{1}{2\sqrt{t}}\bigg)\bigg(4+\frac{1}{\varepsilon}\bigg)\frac{m\sqrt{s}}{\rho_{ts}^2}\eta\nonumber\\
&=\tau\frac{2\|{\bm  x}_{-\max(s)}\|_1}{\sqrt{s}}+{\alpha\tau}\frac{\|{\bm  h}\|_2}{\sqrt{s}}
+\bigg(4+\frac{1}{\varepsilon}\bigg)\frac{(2\sqrt{t}+1)m\sqrt{s}}{2\sqrt{t}\rho_{ts}^2}\eta,
\end{align}
where the last equality is from
$$
\tau=\frac{1}{2\sqrt{t}}\Bigg((2\sqrt{t}+1)\bigg(\frac{2(1+\delta_{ts}^{ub})}{a(s,ts;\alpha)\rho_{ts}}
+\frac{\varepsilon}{2}\bigg)+1\Bigg)
$$
with
$$
\varepsilon=\frac{(1+\delta_{ts}^{ub})}{4a(s,ts;\alpha)\rho_{ts}}.
$$
Then
\begin{eqnarray}
\frac{\alpha\tau}{\sqrt{s}}=\frac{1}{2\sqrt{t}}\Bigg((2\sqrt{t}+1)\frac{17(1+\delta_{ts}^{ub})}{8a(s,t;\alpha)\rho_{s,t}}+1\Bigg)\frac{\alpha}{\sqrt{s}}
<1,
\end{eqnarray}
where the inequality is from \eqref{RIPCondition1}.
In fact,
\begin{align*}
\tau-\frac{\sqrt{s}}{\alpha}
&=\frac{1}{2\sqrt{t}}\Bigg((2\sqrt{t}+1)\frac{17(1+\delta_{ts}^{ub})}{8a(s,ts;\alpha)\rho_{ts}}+1\Bigg)
-\frac{\sqrt{s}}{\alpha}\\
&=\frac{17(2\sqrt{t}+1)}{16\sqrt{t}a(s,ts;\alpha)\rho_{ts}}
\Bigg(1+\delta_{ts}^{ub}-
\bigg(\frac{2\sqrt{ts}}{\alpha}-1\bigg)\frac{8}{17(2\sqrt{t}+1)}a(s,ts;\alpha)\rho_{ts}\Bigg)\\
&=\frac{17(2\sqrt{t}+1)}{16\sqrt{t}a(s,ts;\alpha)\rho_{ts}}
\Bigg(1+\delta_{ts}^{ub}-a(s,t;\alpha)b(s,ts;\alpha)\rho_{ts}\Bigg),
\end{align*}
where
$$
b(t,ts;\alpha)=\bigg(\frac{2\sqrt{ts}}{\alpha}-1\bigg)\frac{8}{17(2\sqrt{t}+1)}
=\frac{8(2\sqrt{ts}-\alpha)}{17\alpha(2\sqrt{t}+1)}.
$$
Then, by
$$
\rho_{ts}=1-\delta_{(t+1)s}^{lb}-\frac{(1+\delta_{ts}^{ub})}{a(s,ts;\alpha)}
$$
one has that
\begin{align*}
\tau-\frac{\sqrt{s}}{\alpha}&=\frac{17(2\sqrt{t}+1)}{16\sqrt{t}a(s,t;\alpha)\rho_{ts}}
\bigg((b(t,ts;\alpha)+1)\delta_{ts}^{ub}+a(s,ts;\alpha)b(t,ts;\alpha)\delta_{(t+1)s}^{lb}\\
&-\Big(a(s,ts;\alpha)b(t,ts;\alpha)-b(t,ts;\alpha)-1\Big)\bigg)<0.
\end{align*}
where the inequality is due to \eqref{RIPCondition1}.
 Therefore,  combing with \eqref{ge:h} and the fact $\tau-\frac{\sqrt{s}}{\alpha}<0$, one has that
\begin{align*}
\|{\bm  h}\|_2&\leq\frac{\sqrt{s}\tau}{\sqrt{s}-\alpha\tau}\frac{2\|\bm{x}_{-\max(s)}\|_1}{\sqrt{s}}\\
&\quad+\frac{2(2\sqrt{t}+1)\big((1+\delta_{ts}^{ub})+a(s,ts;\alpha)\rho_{ts}\big)ms}
{\sqrt{t}(\sqrt{s}-\alpha\tau)(1+\delta_{ts}^{ub})\rho_{ts}^2}\eta,
\end{align*}
which finishes the proof of Theorem \ref{StableRecoveryviaVectorL1-alphaL2-DS}.

\end{proof}

\section{Stable Recovery  Under (L2,~L2)-RIP Frame}\label{s3}

In the section, we develop sufficient conditions
based on the high order  $(\ell_2,\ell_2)$-RIP frame of the
$\ell_{1}-\alpha \ell_2$-DS \eqref{VectorL1-alphaL2-DS}
for the signal recovery applying the technique of the convex combination.

\subsection{Auxiliary Lemmas Under (L2,~L2)-RIP Frame}\label{s3.1}
We first give two auxiliary results under $(\ell_2,\ell_2)$-RIP frame.
The following lemma describes
   a convex combination of sparse vectors for any point based on $\|\cdot\|_1-\|\cdot\|_2$.
   It is developed for  the analysis of the constrained $\ell_{1}-\alpha \ell_2$ minimization and establish improved
high order RIP conditions  for the  signal  recovery.

\begin{lemma}\label{e:convexl1-2}\cite[Lemma2.2]{ge2021new}
Let a vector $\bm {\nu}\in \mathbb{R}^n$ satisfy $\|\bm {\nu}\|_{\infty}\leq \theta$, where $\alpha$ is a positive constant.
Suppose $\|\bm {\nu}\|_{1-2}\leq (s-\sqrt{s})\theta$
 with   a positive integer $s$ and $s\leq |\mathrm{supp}(\bm {\nu})|$.  Then $\bm {\nu}$ can be represented as a convex
combination of $s$-sparse vectors $\bm {u}^{(i)}$, i.e.,
\begin{align}\label{e:com}
\bm {\nu}=\sum_{i=1}^N\lambda_i\bm {u}^{(i)},
\end{align}
where $N$ is a positive integer,
 \begin{align}\label{e:thetaisum}
 0<\lambda_i\leq1,\ \ \ \  \sum_{i=1}^N\lambda_i=1,
 \end{align}
\begin{align}\label{e:setS}
\mathrm{supp}(\bm {u}^{(i)})\subseteq \mathrm{supp}(\bm {\nu}), \ \ \|\bm {u}^{(i)}\|_0\leq {s},\ \ \|\bm {u}^{(i)}\|_{\infty}\leq \Big(1+\frac{\sqrt{2}}{2}\Big)\theta,
\end{align}
and
\begin{align}\label{e:newinequality}
\sum_{i=1}^N\lambda_i\|\bm {u}^{(i)}\|_2^2
 \leq \Big[\Big(1+\frac{\sqrt{2}}{2}\Big)^2(s-\sqrt{s})+1\Big]\theta^2.
\end{align}
\end{lemma}

Combining the above lemma with Lemma \ref{ConeconstraintinequalityforL1-L2},
we next introduce the following new lemma which will
play a crucial role in establishing the recovery condition based on the $L_2$-RIP frame.

\begin{lemma}\label{lem:com}
Assume that $\|\hat{\bm { x}}\|_{\alpha,1-2}\leq\|\bm { x}\|_{\alpha,1-2}$.
Let $\bm { h}=\hat{\bm { x}}-\bm { x}$ and
\begin{align}\label{e:chi}
\chi=
\frac{\sqrt{s}+\alpha}{\sqrt{s}-1}\frac{\|{\bm  h}_{\max(s)}\|_2}{\sqrt{s}}+\frac{2}{s-\sqrt{s}}\|{\bm  x}_{-\max(s)}\|_1
\end{align}
where $s\geq 2$ is a positive integer. And
define  two index sets
\begin{align}
\label{e:W1}
W_1=\left\{i:|\bm {h}_{-\max(s)}(i)|>\frac{\chi}{t-1}\right\},
\end{align}
and
\begin{align}
W_2=\left\{i:|\bm {h}_{-\max(s)}(i)|\leq\frac{\chi}{t-1}\right\},
\label{e:W2}
 \end{align}
where $t=2$ or $t\geq 3$.
 Then the vector  $\bm {h}_{W_2}$ can be represented as a convex
combination of $\lceil ts \rceil-s-W_1$-sparse vectors $\bm {u}^{(i)}$, i.e.,
\begin{align}\label{e:SRnu}
\bm {h}_{W_2}=\sum_{i=1}^N\lambda_i\bm {u}^{(i)},
\end{align}
where $N$ is a positive integer.
And
\begin{align}\label{e:ge3}
\sum_{i=1}^N\lambda_i\|\bm {u}^{(i)}\|_2^2
\leq \frac{\Big(1+\frac{\sqrt{2}}{2}\Big)^2\Big(\lceil ts\rceil-s-\sqrt{\lceil ts\rceil-s}\Big)+1}{(t-1)^2}\chi^2.
\end{align}
\end{lemma}

\begin{proof}
 From  $\|\hat{{\bm  x}}\|_{\alpha,1-2}\leq \|{\bm  x}\|_{\alpha,1-2}$ and  \eqref{e:h-maxsupperbound2} in Lemma \ref{ConeconstraintinequalityforL1-L2}, 
it follows that
\begin{align}\label{Coneconstraintinequality}
\|{\bm  h}_{-\max(s)}\|_1-\|{\bm  h}_{-\max(s)}\|_2
&\leq
\|{\bm  h}_{-\max(s)}\|_1-\alpha\|{\bm  h}_{-\max(s)}\|_2\nonumber\\
&\leq\|{\bm  h}_{\max(s)}\|_1+2\|{\bm  x}_{-\max(s)}\|_1
+\alpha\|{\bm  h}_{\max(s)}\|_2\nonumber\\
&\leq (\sqrt{s}+\alpha)\|{\bm  h}_{\max(s)}\|_2+2\|{\bm  x}_{-\max(s)}\|_1,
\end{align}

where the first inequality comes from $0<\alpha\leq 1$, and the last inequality is because of $\|{\bm  h}_{\max(s)}\|_1\leq \sqrt{s}\|{\bm  h}_{\max(s)}\|_2$.

 We  move to develop  a sparse decomposition for $\bm {h}_{-\max(s)}$ by  Lemma \ref{e:convexl1-2}.
Now, we derive that  $\bm {h}_{-\max(s)}$ satisfies  conditions in Lemma \ref{e:convexl1-2}.

 By \eqref{Coneconstraintinequality}, we have that
\begin{align}\label{e:l1-2upperbound}
\|{\bm  h}_{-\max(s)}\|_1-\|{\bm  h}_{-\max(s)}\|_2
&\leq
\|{\bm  h}_{-\max(s)}\|_1-\alpha\|{\bm  h}_{-\max(s)}\|_2\nonumber\\
\leq&(s-\sqrt{s})\bigg(
\frac{\sqrt{s}+\alpha}{\sqrt{s}-1}\frac{\|{\bm  h}_{\max(s)}\|_2}{\sqrt{s}}+\frac{2}{s-\sqrt{s}}\|{\bm  x}_{-\max(s)}\|_1\bigg)\nonumber\\
=&(s-\sqrt{s})\chi,
\end{align}
where the equality is from the definition of $\chi$ in \eqref{e:chi}.
Using the fact that  $\frac{\sqrt{s}+\alpha}{\sqrt{s}-1}>1$, one has that
\begin{align}\label{e:etainftynew}
\|\bm {h}_{-\max{(s)}}\|_{\infty}\leq \frac{\|\bm {h}_{\max{(s)}}\|_{1}}{s}\leq\frac{\|\bm {h}_{\max{(s)}}\|_{2}}{\sqrt{s}}\leq\frac{\sqrt{s}+\alpha}{\sqrt{s}-1}\frac{\|{\bm  h}_{\max(s)}\|_2}{\sqrt{s}}\leq \chi,
\end{align}
where the last inequality is due to the definition of $\chi$ in \eqref{e:chi}.

Next, we prove  results in the lemma. The following two cases need to be considered.

 $i)$ $t=2$:
 By \eqref{e:etainftynew} and the definition of $W_1$ in \eqref{e:W1}, it is clear that $W_1=\emptyset$.
Then proving  \eqref{e:SRnu} and \eqref{e:ge3} are respectively  equivalent to showing
\begin{align}\label{e:etamaxkresolve1}
 \bm {h}_{-\max(s)}=\sum_{i=1}^N\lambda_i\bm {u}^{(i)},
\end{align}
and
\begin{align}\label{e:ge1}
&\sum_{i=1}^N\lambda_i\|\bm {u}^{(i)}\|_2^2
\leq \Big(\Big(1+\frac{\sqrt{2}}{2}\Big)^2(s-\sqrt{s})+1\Big)\chi^2,
\end{align}
where $0<\lambda_i\leq1$, $\sum_{i=1}^N\lambda_i=1$, and $\|\bm {u}^{(i)}\|_0\leq s$.
By \eqref{e:l1-2upperbound}, \eqref{e:etainftynew} and Lemma \ref{e:convexl1-2},
it is clear that \eqref{e:etamaxkresolve1} and \eqref{e:ge1} hold.
 Thus we completes the proofs of \eqref{e:SRnu} and \eqref{e:ge3} for $t=2$.

$ii)$ $t\geq 3$: We prove \eqref{e:SRnu} and \eqref{e:ge3} by applying
Lemma \ref{e:convexl1-2}.
Then, we first  consider upper bounds of  $\|\bm {h}_{W_2}\|_\infty$ and $\|\bm {h}_{W_2}\|_1-\|\bm {h}_{W_2}\|_2$.
By definitions of  $W_1$  and $W_2$, one has
\begin{align*}
W_1\cap W_2=\emptyset,\ \ \ \ \bm {h}_{-\max(s)}=\bm {h}_{W_1}+\bm {h}_{W_2}
\end{align*}
and
\begin{align}\label{e:etaW2infty}
\|\bm {h}_{W_2}\|_{\infty}\leq\frac{\chi}{t-1}.
\end{align}

 We next establish  a upper bound of $\|\bm {h}_{W_2}\|_1-\|\bm {h}_{W_2}\|_2$.
From  $\bm {h}_{-\max(s)}=\bm {h}_{W_1}+\bm {h}_{W_2}$, it follows that
\begin{align}\label{e:etaW2}
\|\bm {h}_{W_2}\|_{1}-\|\bm {h}_{W_2}\|_{2}
=&\|\bm {h}_{-\max{(s)}}-\bm {h}_{W_1}\|_{1}
-\|\bm {h}_{-\max{(s)}}-\bm {h}_{W_1}\|_{2}\nonumber\\
\overset{(a)}{=}&\|\bm {h}_{-\max{(s)}}\|_1-\|\bm {h}_{W_1}\|_{1}
-\|\bm {h}_{-\max{(s)}}-\bm {h}_{W_1}\|_{2}\nonumber\\
\overset{(b)}{\leq}&\|\bm {h}_{-\max{(s)}}\|_1-\|\bm {h}_{-\max{(s)}}\|_2
-(\|\bm {h}_{W_1}\|_{1}-\|\bm {h}_{W_1}\|_{2})\nonumber\\
\overset{(c)}{\leq}&(s-\sqrt{s})\chi-(\|\bm {h}_{W_1}\|_{1}-\|\bm {h}_{W_1}\|_{2}),
\end{align}
where $(a)$, $(b)$ and $(c)$  follow from  $W_1\subseteq \mathrm{supp}(\bm {h}_{-\max(s)})$,  the triangle inequality on $\|\cdot\|_2$, and \eqref{e:l1-2upperbound}, respectively.
 Then, establishing  a upper bound of $\|\bm {h}_{W_2}\|_1-\|\bm {h}_{W_2}\|_2$ is equivalent to showing
 the lower bound of $\|\bm {h}_{W_1}\|_{1}-\|\bm {h}_{W_1}\|_{2}$.
 By Lemma \ref{LocalEstimateL1-L2} (b) with $S_1=W_1$ and $S_2=W_2$, we have that
\begin{align}\label{e:eta-maxkuplowbound}
\|\bm {h}_{-\max(s)}\|_{1}-\|\bm {h}_{-\max(s)}\|_{2}
\geq&\big(\|\bm {h}_{W_1}\|_{1}-\|\bm {h}_{W_1}\|_{2}\big)
+\big(\|\bm {h}_{W_2}\|_{1}-\|\bm {h}_{W_2}\|_{2}\big)\nonumber\\
\geq&\big(\|\bm {h}_{W_1}\|_{1}-\|\bm {h}_{W_1}\|_{2}\big)\nonumber\\
\overset{(1)}{\geq}&(|W_1|-\sqrt{|W_1|})\min_{i\in W_1}|\bm {h}_{W_1}(i)|\nonumber\\
\overset{(2)}{\geq}&(|W_1|-\sqrt{|W_1|})\frac{\chi}{t-1},
\end{align}
where $(1)$ and $(2)$ follow from \eqref{e:l12alphas} in Lemma \ref{LocalEstimateL1-L2} (a) and  the definition of $W_1$ in \eqref{e:W1},
respectively.
Combining the above   inequality with
\eqref{e:etaW2},  we have
\begin{align}\label{e:l1-2upperbounds}
\|\bm {h}_{W_2}\|_{1}-\|\bm {h}_{W_2}\|_{2}\leq \Big((s(t-1)-|W_1|)-
(\sqrt{s}(t-1)-\sqrt{|W_1|})\Big)\frac{\chi}{t-1}.
\end{align}

Inspired by Lemma \ref{e:convexl1-2}, \eqref{e:etaW2infty} and \eqref{e:l1-2upperbounds},
we  explore that
\begin{align}\label{e:etaW2upperbound}
\|\bm {h}_{W_2}\|_{1}-\|\bm {h}_{W_2}\|_{2}\leq
\Big(s(t-1)-|W_1|-\sqrt{s(t-1)-|W_1|}\Big)\frac{\chi}{t-1}.
\end{align}
In fact, based on \eqref{e:l1-2upperbounds}, \eqref{e:etaW2upperbound} follows from
\begin{itemize}
\item[(i)] $|W_1|< s(t-1)$,
\item[(ii)]
$\sqrt{s(t-1)-|W_1|}\leq \sqrt{s}(t-1)-\sqrt{|W_1|}$
\end{itemize}
for   $t\geq3$ and $s\geq 2$.  We first prove item (i).
Observe that
\begin{align*}
(|W_1|-\sqrt{|W_1|})\frac{\chi}{t-1}
&\leq\|\bm {h}_{-\max(s)}\|_{1}-\|\bm {h}_{-\max(s)}\|_{2}\nonumber\\
&\leq\Big(s(t-1)-\sqrt{s}(t-1)\Big)\frac{\chi}{t-1}
\end{align*}
where the equalities follow from \eqref{e:eta-maxkuplowbound} and
\eqref{e:l1-2upperbound}, respectively.
Furthermore, since $t\geq 3$, which  implies
$\sqrt{s}(t-1)> \sqrt{s(t-1)}$, it is clear that
$$|W_1|-\sqrt{|W_1|}< s(t-1)-\sqrt{s(t-1)},$$
meaning
\begin{align}\label{e:W1upperbound}
|W_1|<s(t-1),
\end{align}
since the function $g(x)=x-\sqrt{x}$ is increasing  monotony for $x\geq1$ when $|W_1|\geq 1$,
and $s(t-1)\geq 2s \geq4$ using $t\geq3$ and $s\geq2$ when $|W_1|=0$.

Now, we  turn to show item (ii). We observe that the second-order function
\begin{align*}
f(\sqrt{|W_1|})=&(\sqrt{s(t-1)-|W_1|})^2-(\sqrt{s}(t-1)-\sqrt{|W_1|})^2\\
=&-2|W_1|+2\sqrt{s|W_1|}(t-1)+s(t-1)(2-t)\leq 0,
\end{align*}
which implies item (ii), i.e., $\sqrt{s(t-1)-|W_1|}\leq \sqrt{s}(t-1)-\sqrt{|W_1|}$,
 since
its discriminant $\Delta=4s(t-1)^2+8s(t-1)(2-t)=4s(t-1)(-t+3)\leq 0$
under $t\geq3$.

Then, from \eqref{e:etaW2infty}, \eqref{e:etaW2upperbound} and Lemma \ref{e:convexl1-2} with  $\theta=\frac{\chi}{t-1}$,
$\bm {\nu}=\bm {h}_{W_2}$ and $s=\lceil ts\rceil-s-|W_1|$,
where $s\geq 1$ since $W_1>s(t-1)$,
it follows that
\begin{align}\label{e:etamaxkresolvehigh}
 \bm {h}_{W_2}=\sum_{i=1}^N\lambda_i\bm {u}^{(i)},
\end{align}
where $N$ is certain positive integer,  $0<\lambda_i\leq1$ with $\sum_{i=1}^N\lambda_i=1$.
And using \eqref{e:newinequality}in Lemma \ref{e:convexl1-2}, we have
\begin{align}\label{e:omegal2high}
&\sum_{i=1}^N\lambda_i\|\bm {u}^{(i)}\|_2^2
\leq \Big[\Big(1+\frac{\sqrt{2}}{2}\Big)^2(\lceil ts\rceil-s-|W_1|-\sqrt{\lceil ts\rceil-s-|W_1|})+1\Big]\Big(\frac{\chi}{t-1}\Big)^2\nonumber\\
&\leq \Big[\Big(1+\frac{\sqrt{2}}{2}\Big)^2(\lceil ts\rceil-s-\sqrt{\lceil ts\rceil-s})+1\Big]\Big(\frac{\chi}{t-1}\Big)^2\nonumber\\
&=\frac{\Big(1+\frac{\sqrt{2}}{2}\Big)^2\Big(\lceil ts\rceil-s-\sqrt{\lceil ts\rceil-s}\Big)+1}{(t-1)^2}\chi^2,
\end{align}
where the second  inequality follows from  the decreasing monotony of the function
$g(x)=x-\sqrt{x}$ for $x\geq 1$ and $\lceil ts\rceil-s-|W_1|\geq 1$.
We complete the proof.
\end{proof}

\subsection{Main Result Under (L2,~L2)-RIP Frame }\label{s3.2}

Now, we show the stable recovery under $(\ell_2,\ell_2)$-RIP frame.

\begin{theorem}\label{New-StableRecoveryviaVectorL1-alphaL2-DS}
Consider $\bm { b}=\bm { Ax}+\bm {e}$ with $\|\bm { A}^{\top}\bm {e}\|_\infty\leq \eta$.
Let $\hat{\bm {x}}^{DS}$ be the minimizer
of the  $\ell_{1}-\alpha \ell_{2}$-DS (\ref{VectorL1-alphaL2-DS}). If the  measurement matrix $\bm {A}$ satisfies
\begin{equation}\label{RIPCondition2}
\delta_{ts}<\frac{1}{\sqrt{\frac{(\sqrt{s}+\alpha)^2\Big(\Big(1+\frac{\sqrt{2}}{2}\Big)^2( ts-s-\sqrt{ts-s})+1\Big)}{s(t-1)^2(\sqrt{s}-1)^2}+1}},
\end{equation}
for some  $t\geq3$ or $t=2$, where $s\geq 2$ is an  positive integer.
Then
\begin{align*}
\|\hat{{\bm x}}^{DS}-{\bm x}\|_2
\leq&\Bigg(1+\sqrt{\frac{{\alpha}+\sqrt{s}}{\sqrt{s}}+\frac{\alpha^2}{4s}}
+\frac{\alpha+\sqrt{2}}{2\sqrt{s}}\Bigg)\frac{2(1-\mu)\mu{\sqrt{\lceil ts\rceil}}}{\mu-\mu^2-(\mu-1)^2\delta_{ts}}\eta\nonumber\\
&+{\Bigg[\bigg(\sqrt{s}+\sqrt{{\alpha\sqrt{s}}+s+\frac{\alpha^2}{4}}
+\frac{\alpha+\sqrt{2}}{2}\bigg)
\Bigg(\frac{2\delta_{ts}(1-2\mu)}
{(\sqrt{s}+\alpha)\big(\mu-\mu^2-(\mu-1)^2\delta_{ts}\big)}}\nonumber\\
&+
{\sqrt{\frac{2\delta_{ts}(1-2\mu)}{(\sqrt{s}+\alpha)^2(\mu-\mu^2-(\mu-1)^2\delta_{ts})}}\Bigg)
+\frac{\sqrt{2s}}{2}\Bigg]
\frac{\|\bm {x}_{-\max(s)}\|_1}{\sqrt{s}}},
\end{align*}
where
\begin{align}
\label{e:mu}
\mu=\frac{1}{\sqrt{1+\frac{(\sqrt{s}+\alpha)^2\bigg(\Big(1+\frac{\sqrt{2}}{2}\Big)^2\big(\lceil ts\rceil-s-\sqrt{\lceil ts\rceil-s}\big)+1\bigg)}{s(t-1)^2(\sqrt{s}-1)^2}}+1}.
\end{align}
\end{theorem}
\begin{remark}
When $t=2$, the condition \eqref{RIPCondition2} 
reduces to
\begin{equation*}
\delta_{2s}<\frac{1}{\sqrt{1+\frac{(\sqrt{s}+{\alpha})^2\Big(\Big(1+\frac{\sqrt{2}}{2}\Big)^2( s-\sqrt{s})+1\Big)}{s(\sqrt{s}-1)^2}}}.
\end{equation*}
It is clearly  weaker than the following  condition in \cite[Theorem 3.4]{ge2021new}
\begin{equation*}
\delta_{2s}<\frac{1}{\sqrt{1+\frac{(\sqrt{s}+1)^2\Big(\Big(1+\frac{\sqrt{2}}{2}\Big)^2( s-\sqrt{s})+1\Big)}{s(\sqrt{s}-1)^2}}}
\end{equation*}
for the  $\ell_{1}-\alpha \ell_{2}$-DS \eqref{VectorL1-alphaL2-DS} with $\alpha=1$, because of $0<\alpha\leq 1$.
\end{remark}

\begin{remark}
When {$s\geq2$} and $\alpha=1$,
 by the condition \eqref{RIPCondition2}, 
one has
\begin{align}\label{RIPConditionge}
\delta_{ts}<&\frac{1}{\sqrt{1+\frac{(\sqrt{s}+1)^2\Big(\Big(1+\frac{\sqrt{2}}{2}\Big)^2( ts-s-\sqrt{ts-s})+1\Big)}{s(t-1)^2(\sqrt{s}-1)^2}}}\nonumber\\
=&\frac{1}{\sqrt{1+\frac{(\sqrt{s}+1)^2(2+\sqrt{2})^2}{2(\sqrt{s}-1)^2(t-1)}\cdot\frac{ts-s-\sqrt{ts-s}
+\Big(\frac{2}{2+\sqrt{2}}\Big)^2}{s(t-1)}}}
\end{align}
It is   weaker than \cite[the condition (3.30)]{ge2021new}
\begin{equation*}
\delta_{ts}<\frac{1}{\sqrt{1+\frac{(\sqrt{s}+1)^2(2+\sqrt{2})^2}{2(\sqrt{s}-1)^2(t-1)}}},
\end{equation*}
since $\sqrt{ts-s}>1$ and $\Big(\frac{2}{2+\sqrt{2}}\Big)^2<1$ implying  $\frac{ts-s-\sqrt{ts-s}+\Big(\frac{2}{2+\sqrt{2}}\Big)^2}{s(t-1)}<1$.
Furthermore, the condition \eqref{RIPCondition2} with $0<\alpha<1$ is  weaker than \eqref{RIPConditionge}.
\end{remark}

\begin{proof}
Our proof is motivated by the proof of \cite[Theorem 3.1]{ge2021new}.
Since  $\bm {\hat{x}}^{DS}$ is the minimizer  of \eqref{VectorL1-alphaL2-DS},
 $\|\bm { \hat{x}}^{DS}\|_{\alpha,1-2}\leq \|{\bm {x}}\|_{\alpha,1-2}$
and $\|\bm { A}^{\top}(\bm {b}-\bm {A}\hat{\bm { x}}^{DS})\|_\infty\leq \eta$.

Take $\bm {h}=\bm {\hat{ x}}^{DS}-\bm {x}$.
On the one hand,  using  $\|{\bm  A}^{\top}\bm { e}\|_\infty=\|\bm {A}^{\top}(\bm {b}-\bm {Ax})\|_\infty\leq \eta$
and $\|{\bm  A}^{\top}({\bm  b}-{\bm  A}\hat{{\bm  x}}^{DS})\|_\infty\leq \eta$,
we have the following tube constraint inequality
\begin{align}\label{e:Tubeconstraintinequality}
\|{\bm  A}^{\top}\bm  {Ah}\|_{\infty}
&=\|{\bm  A}^{\top}(\bm  A {\hat{{\bm  x}}}^{DS}-\bm { Ax)}\|_{\infty}\nonumber\\
&\leq\|\bm { A}^{\top}(\bm { A}\hat{\bm { x}}^{DS}-\bm  {b})\|_{\infty}+\|\bm { A}^{\top}(\bm { b}-\bm { Ax})\|_{\infty}\nonumber\\
&\leq\eta+\eta=2\eta.
\end{align}

On the other hand,
using $\|\bm { \hat{x}}^{DS}\|_{\alpha,1-2}\leq \|{\bm {x}}\|_{\alpha,1-2}$ and Lemma \ref{lem:com}, one has that
\begin{align}\label{ee:SRnu}
\bm {h}_{W_2}=\sum_{i=1}^N\lambda_i\bm {u}^{(i)},
\end{align}
and
\begin{align*}
\sum_{i=1}^N\lambda_i\|\bm {u}^{(i)}\|_2^2
\leq \frac{\Big(1+\frac{\sqrt{2}}{2}\Big)^2(\lceil ts\rceil-s-\sqrt{\lceil ts\rceil-s})+1}{(t-1)^2}\chi^2,
\end{align*}
where $0<\lambda_i\leq 1$, $\sum_{i=1}^N\lambda_i=1$,
$\bm {u}^{(i)}$ ($1 \leq i\leq N$) are  $(\lceil ts \rceil-s-W_1)$-sparse vectors and
$W_1, \ W_2, \chi $  are respectively defined in \eqref{e:W1}, \eqref{e:W2} and \eqref{e:chi}.
By the definition of $\chi$ in \eqref{e:chi}, we get
\begin{align}\label{e:l2uibound}
\sum_{i=1}^N\lambda_i\|\bm {u}^{(i)}\|_2^2
\leq&\frac{\Big(1+\frac{\sqrt{2}}{2}\Big)^2(\lceil ts\rceil-s-\sqrt{\lceil ts\rceil-s})+1}{(t-1)^2}\bigg( \frac{(\sqrt{s}+\alpha)^2}{s(\sqrt{s}-1)^2}\|{\bm  h}_{\max(s)}\|_2^2\nonumber\\
&+\frac{4(\sqrt{s}+\alpha)}{s(\sqrt{s}-1)^2}\|{\bm  h}_{\max(s)}\|_2\|{\bm  x}_{-\max(s)}\|_1+
\frac{4}{s(\sqrt{s}-1)^2}\|{\bm  x}_{-\max(s)}\|_1^2\bigg)\nonumber\\
=&\frac{1-2\mu}{\mu^2}\bigg(\|{\bm  h}_{\max(s)}\|_2^2+\frac{4\|{\bm  h}_{\max(s)}\|_2\|{\bm  x}_{-\max(s)}\|_1}{\sqrt{s}+\alpha} +\frac{4\|{\bm  x}_{-\max(s)}\|_1^2}{(\sqrt{s}+\alpha)^2}\bigg),
\end{align}
where the equality is from the definition of $\mu$ in \eqref{e:mu}.

Next, we develop the inequality on $\|\bm {h}_{\max(s)}+\bm {h}_{W_1}\|_2$ to estimate an upper bound of
$\|\bm {h}_{\max(s)}+\bm {h}_{W_1}\|_2$ by the following identity.
\begin{align}\label{e:identity}
&\sum_{i=1}^N\frac{\lambda_{i}}{4}
\|\bm {A}(\bm {h}_{\max{(s)}}+\bm {h}_{W_1}+\mu \bm {u}^{(i)})\|_{2}^{2}\nonumber\\
&=\sum\limits_{i=1}^{N}\lambda_{i}
\Big\|\bm {A}\Big((\frac{1}{2}-\mu)
(\bm {h}_{\max(s)}+\bm {h}_{W_1})-\frac{1}{2}\mu \bm {u}^{(i)}+\mu \bm {h}\Big)\Big\|_{2}^{2}.
\end{align}

We show an upper bound on the right-hand side (RHS) and a lower bound on the left-hand side (LHS) for the  identity
\eqref{e:identity}
  by following  techniques for deriving   \cite[the inequality (3.49) and (3.52)]{ge2021new}, respectively.
 They are not hard to check that
\begin{align*}
\text{LHS}
\geq&(1-\delta_{\lceil ts \rceil})
\sum_{i=1}^N\frac{\lambda_{i}}{4}
\|\bm {h}_{\max{(s)}}+\bm {h}_{W_1}+\mu \bm {u}^{(i)}\|_{2}^{2}\nonumber\\
=&\frac{1-\delta_{\lceil ts \rceil}}{4}
\Big(\|\bm {h}_{\max{(s)}}+\bm {h}_{W_1}\|_2^2+\mu^2\sum_{i=1}^N\lambda_{i} \|\bm {u}^{(i)}\|_{2}^{2}\Big),
\end{align*}
where the inequality is due to that $\bm {h}_{\max{(s)}}+\bm {h}_{W_1}+\mu \bm {u}^{(i)}$  is $ts$-sparse, and
the equality follows from  $\sum_{i=1}^N\lambda_i=1$, $\mathrm{supp}(\bm {u}^{(i)})\subseteq \mathrm{supp}(\bm {h}_{W_2})\subseteq \mathrm{supp}(\bm {h}_{-\max(s)})$ and $W_1\cap W_2=\emptyset$.
For \text{RHS} of  the identity  \eqref{e:identity}, we have that
\begin{align*}
&\text{RHS}\nonumber\\
&=\sum\limits_{i=1}^{N}\lambda_{i}
\Big\|\bm {A}\Big((\frac{1}{2}-\mu)(\bm {h}_{\max(s)}+\bm {h}_{W_1})
-\frac{1}{2}\mu \bm {u}^{(i)}\Big)\Big\|_{2}^{2}+(1-\mu)\mu\langle \bm {A} (\bm {h}_{\max(s)}+\bm {h}_{W_1}), \bm {A h}\rangle\nonumber\\
&\overset{(a)}{\leq}(1+\delta_{\lceil ts \rceil}) \sum\limits_{i=1}^{N}\lambda_{i}
\Big\|(\frac{1}{2}-\mu)(\bm {h}_{\max(s)}+\bm {h}_{W_1})
-\frac{1}{2}\mu \bm {u}^{(i)}\Big\|_{2}^{2}
+(1-\mu)\mu\langle \bm {h}_{\max(s)}+\bm {h}_{W_1},  \bm {A}^{\top}\bm {A h}\rangle\nonumber\\
&\overset{(b)}{\leq}(1+\delta_{\lceil ts \rceil})
\Big((\frac{1}{2}-\mu)^2\|\bm {h}_{\max(s)}+\bm {h}_{W_1}\|_2^2
+\frac{\mu^2}{4}\sum_{i=1}^N\lambda_i
\|\bm {u}^{(i)}\|_2^2\Big)\nonumber\\
&\hspace*{12pt}+2(1-\mu)\mu\sqrt{\lceil ts \rceil}\eta
\|\bm {h}_{\max(s)}+\bm {h}_{W_1}\|_2,
\end{align*}
where  $(a)$  is from that $\bm {A}$ satisfies the $ts$-order RIP and  $(\frac{1}{2}-\mu)(\bm {h}_{\max(s)}+\bm {h}_{W_1})
-\frac{1}{2}\mu \bm {u}^{(i)}$ is $ts$-sparse, and
$(b)$ is because of   $\sum_{i=1}^N\lambda_i=1$, $\mathrm{supp}(\bm {u}^{(i)})\subseteq \mathrm{supp}(\bm {h}_{W_2})\subseteq \mathrm{supp}(\bm {h}_{-\max(s)})$, $W_1\cap W_2=\emptyset$, the Cauchy-Schwarz inequality
with $|\text{supp}(\bm {h}_{\max(s)}+\bm {h}_{W_1})|=s+|W_1|$ and the fact $|W_1|< s(t-1)$ in \eqref{e:W1upperbound} and \eqref{Tubeconstraintinequality}.

Combining the identity \eqref{e:identity} with the above two  inequalities, we have
\begin{align*}
&\Big((1+\delta_{\lceil ts \rceil})(\frac{1}{2}-\mu)^2
-\frac{1-\delta_{\lceil ts \rceil}}{4}\Big)\|\bm {h}_{\max(s)}+\bm {h}_{W_1}\|_2^2
+\frac{\mu^2\delta_{\lceil ts \rceil}}{2}
\sum_{i=1}^N\lambda_i\|\bm {u}^{(i)}\|_2^2\\
&+2(1-\mu)\mu\sqrt{\lceil ts \rceil}\eta
\|\bm {h}_{\max(s)}+\bm {h}_{W_1}\|_2\geq 0.
\end{align*}

By applying \eqref{e:l2uibound}, the definition of $\chi$ in \eqref{e:chi} and $\|\bm {h}_{\max(s)}\|_2\leq \|\bm {h}_{\max(s)}+\bm {h}_{W_1}\|_2$,
the above inequality reduces to
\begin{align*}
&\Big[\mu^2-\mu+(\mu-1)^2\delta_{\lceil ts \rceil}\Big]
\|\bm {h}_{\max(s)}+\bm {h}_{W_1}\|_2^2+\Big[2(1-\mu)\mu\sqrt{\lceil ts \rceil}\eta
\nonumber\\&+\frac{2\delta_{\lceil ts \rceil}(1-2\mu)\|\bm {x}_{-\max(s)}\|_1}{\sqrt{s}+\alpha}\Big]\|\bm {h}_{\max(s)}+\bm {h}_{W_1}\|_2
+\frac{2\delta_{\lceil ts \rceil}(1-2\mu)}{(\sqrt{s}+\alpha)^2}\|\bm {x}_{-\max(s)}\|_1^2
\geq0,
\end{align*}
which is a second-order inequality for $\|\bm {h}_{\max(s)}+\bm {h}_{W_1}\|_2$.
Note that
$$\delta_{ts}=\left\{
                \begin{array}{ll}
                  \delta_{ts}, & \hbox{$ts$ is  an integer} \\
                  \delta_{\lceil ts\rceil }, & \hbox{$ts$ is not an integer}
                \end{array}
              \right.
$$
and
$$
h(x)=\frac{1}{\sqrt{\frac{(\sqrt{s}+\alpha)^2\Big(\Big(1+\frac{\sqrt{2}}{2}\Big)^2(x-\sqrt{x})\Big)}{s(t-1)^2(\sqrt{s}-1)^2}+1}}$$
is a decrease function for $x\geq\frac{\sqrt{2}}{2}$.
Thus, our condition \eqref{RIPCondition2} on $\delta_{ts}$ with $t=2$ or $t\geq 3$ and $s\geq 2$ can guarantee that
\begin{align}\label{deltatkupper}
\delta_{ts}<\frac{\mu}{1-\mu}=\frac{1}{\sqrt{\frac{(\sqrt{s}+\alpha)^2\Big(\Big(1+\frac{\sqrt{2}}{2}\Big)^2(\lceil ts\rceil-s-\sqrt{\lceil ts\rceil-s})+1\Big)}{s(t-1)^2(\sqrt{s}-1)^2}+1}}
\end{align}
holds.
 By solving the above second-order inequality
under \eqref{deltatkupper} we get
\begin{align}\label{e:etamaxkboundhigh}
&\|\bm {h}_{\max(s)}+\bm {h}_{W_1}\|_2\nonumber\\
&\leq \frac{2(1-\mu)\mu\sqrt{\lceil ts\rceil}\eta+\Big(\frac{2\delta_{ts}(1-2\mu)}{\sqrt{s}+\alpha}+
\sqrt{\frac{2\delta_{ts}(1-2\mu)}{(\sqrt{s}+\alpha)^2}(\mu-\mu^2-(\mu-1)^2\delta_{ts})}\Big)\|\bm {x}_{-\max(s)}\|_1}{\mu-\mu^2-(\mu-1)^2\delta_{ts}},
\end{align}
 where is from $z\leq \frac{b +\sqrt{b^2+4ac }}{2a}\leq \frac{b +\sqrt{ac }}{a}$ satisfying second-order inequality
$az^2-bz-c\leq 0$ with $a,b,c>0$.

Following the argument in \cite[ Step 2]{ge2021new} and \eqref{Coneconstraintinequality},
we can express an upper bound  $\|\bm {h}_{-\max(s)}\|_2$. First, by estimation \eqref{e:h-maxsupperbound1} in Lemma \ref{ConeconstraintinequalityforL1-L2} and $\|\bm {h}_{-\max{(s)}}\|_{\infty}\leq \|\bm {h}_{\max{(s)}}\|_{1}/s\leq\|\bm {h}_{\max{(s)}}\|_{2}/\sqrt{s}$, we know
\begin{align*}
&\|\bm {h}_{-\max{(s)}}\|_2^2\\
&\leq \|\bm {h}_{-\max{(s)}}\|_{1}\|\bm {h}_{-\max{(s)}}\|_{\infty}\nonumber\\
&\leq\big((\sqrt{s}+\alpha)\|\bm {h}_{\max{(s)}}\|_2+\alpha\|\bm {h}_{-\max{(s)}}\|_2+2\|\bm {x}_{-\max{(s)}}\|_1\big)
\frac{\|\bm {h}_{\max{(s)}}\|_2}{\sqrt{s}}\nonumber\\
&=\frac{\alpha\|\bm {h}_{\max{(s)}}\|_2}{\sqrt{s}}\|\bm {h}_{-\max{(s)}}\|_2+\frac{\sqrt{s}+\alpha}{\sqrt{s}}\|\bm {h}_{\max{(s)}}\|_2^2+\frac{2\|\bm {x}_{-\max{(s)}}\|_1\|\bm {h}_{\max{(s)}}\|_2}{\sqrt{s}}.
\end{align*}
Therefore
\begin{align*}
&\bigg(\|\bm {h}_{-\max{(s)}}\|_2-\frac{\alpha\|\bm {h}_{\max{(s)}}\|_2}{2\sqrt{s}}\bigg)^2\\
&\leq\bigg(\frac{\alpha^2}{4s}+\frac{\alpha+\sqrt{s}}{\sqrt{s}}\bigg)\|\bm {h}_{\max{(s)}}\|_2^2
+\frac{2\|\bm {x}_{-\max{(s)}}\|_1\|\bm {h}_{\max{(s)}}\|_2}{\sqrt{s}}
\end{align*}
which implies
\begin{align}\label{e:anotherupperbound}
\|\bm {h}_{-\max{(s)}}\|_2&\leq \Bigg(\sqrt{\frac{\alpha+\sqrt{s}}{\sqrt{s}}+\frac{\alpha^2}{4s}}+\frac{\alpha}{2\sqrt{s}}\Bigg)\|\bm {h}_{\max{(s)}}\|_2
+\sqrt{\frac{2\|\bm {x}_{-\max{(s)}}\|_1\|\bm {h}_{\max{(s)}}\|_2}{\sqrt{s}}}\nonumber\\
&\leq\Bigg(\sqrt{\frac{\alpha+\sqrt{s}}{\sqrt{s}}+\frac{\alpha^2}{4s}}+\frac{\alpha+\sqrt{2}}{2\sqrt{s}}\Bigg)\|\bm {h}_{\max{(s)}}\|_2
+\frac{\sqrt{2}}{2}\|\bm {x}_{-\max{(s)}}\|_1,
\end{align}
where the second inequality  comes from the basic inequality $2|a||b|\leq \frac{(|a|+|b|)^2}{2}$.

Therefore, from the fact $\|\bm {h}\|_2=\sqrt{\|\bm {h}_{\max{(s)}}\|_2^2 +\|\bm {h}_{-\max{(s)}}\|_2^2}$ and the above estimation \eqref{e:anotherupperbound}, it follows that
 \begin{align*}
\|\bm {h}\|_2
\leq&\sqrt{\|\bm {h}_{\max(s)}\|_2^2
+\Bigg[\Big(\sqrt{\frac{\alpha+\sqrt{s}}{\sqrt{s}}+\frac{\alpha^2}{4s}}
+\frac{\alpha+\sqrt{2}}{2\sqrt{s}}\Big)\|\bm {h}_{\max(s)}\|_2
+\frac{\sqrt{2}}{2}\|\bm {x}_{-\max(s)}\|_1\Bigg]^2}\\
\leq&
\Bigg(1+\sqrt{\frac{{\alpha}+\sqrt{s}}{\sqrt{s}}+\frac{\alpha^2}{4s}}
+\frac{\alpha+\sqrt{2}}{2\sqrt{s}}\Bigg)\|\bm {h}_{\max(s)}\|_2
+\frac{\sqrt{2}}{2}\|\bm {x}_{-\max(s)}\|_1\\
\leq&\Bigg(1+\sqrt{\frac{{\alpha}+\sqrt{s}}{\sqrt{s}}+\frac{\alpha^2}{4s}}
+\frac{\alpha+\sqrt{2}}{2\sqrt{s}}\Bigg)\frac{2(1-\mu)\mu{\sqrt{\lceil ts\rceil}}}{\mu-\mu^2-(\mu-1)^2\delta_{ts}}\eta\nonumber\\
&+{\Bigg[\bigg(\sqrt{s}+\sqrt{{\alpha\sqrt{s}}+s+\frac{\alpha^2}{4}}
+\frac{\alpha+\sqrt{2}}{2}\bigg)
\Bigg(\frac{2\delta_{ts}(1-2\mu)}
{(\sqrt{s}+\alpha)\big(\mu-\mu^2-(\mu-1)^2\delta_{ts}\big)}}\nonumber\\
&+
{\sqrt{\frac{2\delta_{ts}(1-2\mu)}{(\sqrt{s}+\alpha)^2(\mu-\mu^2-(\mu-1)^2\delta_{ts})}}\Bigg)
+\frac{\sqrt{2s}}{2}\Bigg]
\frac{\|\bm {x}_{-\max(s)}\|_1}{\sqrt{s}}},
\end{align*}
where the last inequality is due to \eqref{e:etamaxkboundhigh}.
Therefore, we complete the proof.

\end{proof}

\section{Effective Algorithm for L1-L2-DS }\label{s4}

In the section, we present an effective algorithm to  solve the $\ell_1-\alpha \ell_2$-DS  \eqref{VectorL1-alphaL2-DS}. Based on the fact that  Dantzig selector and Lasso estimator exhibit similar behavior,
 we propose an unconstraint penalty problem
 as follows
\begin{eqnarray}\label{VectorL1-alphaL2-DS-penalty}
\min_{{\bm  x},{\bm  y}\in\mathcal{B}^{\infty}(\eta)}&\lambda(\|{\bm  x}\|_1-\alpha\|{\bm  x}\|_2)
+\frac{1}{2}\|{\bm  A}^{\top}{\bm  A}{\bm  x}-{\bm  y}-{\bm  A}^{\top}{\bm  b}\|_2^{2},
\end{eqnarray}
where $\lambda>0$ is the regularized parameter.

We find the optimal solution of \eqref{VectorL1-alphaL2-DS-penalty} using the alternating direction method of multipliers (ADMM) algorithm.
First,  splitting  the term $\|{\bm  x}\|_1-\alpha\|{\bm  x}\|_2$ and letting ${\bm B}={\bm  A}^{\top}{\bm  A},{\bm c}={\bm  A}^{\top}{\bm  b}$, one gets an equivalent problem of \eqref{VectorL1-alphaL2-DS-penalty}:
\begin{eqnarray}\label{VectorL1-alphaL2-DS-penalty-Equivalent}
&&\min_{{\bm  x},{\bm  w},{\bm  y}\in\mathcal{B}^{\infty}(\eta)}\lambda(\|{\bm  w}\|_1-\alpha\|{\bm  w}\|_2)
+\frac{1}{2}\|{\bm  B}{\bm  x}-{\bm  y}-{\bm  c}\|_2^{2},\nonumber\\
&&~\text{s.~t.~}\hspace*{24pt}{\bm  x}-{\bm  w}={\bm  0}.
\end{eqnarray}

The augmented Lagrangian function of \eqref{VectorL1-alphaL2-DS-penalty-Equivalent} is
\begin{eqnarray}\label{ALagrange-PDS}
\mathcal{L}_{\beta}({\bm  x},{\bm  y},{\bm  w};{\bm  z})
=\lambda(\|{\bm  w}\|_1-\alpha\|{\bm  w}\|_2)
+\frac{1}{2}\|{\bm  B}{\bm  x}-{\bm  y}-{\bm  c}\|_2^{2}
+\frac{\beta}{2}\|{\bm  x}-{\bm  w}\|_2^2+\langle {\bm  z}, {\bm  x}-{\bm  w}\rangle,
\end{eqnarray}
where ${\bm  z}$ is the Lagrangian multiplier.
Given $({\bm  x}^0,{\bm y}^0,{\bm w}^0;{\bm  z}^0)$,  
iterations for \eqref{ALagrange-PDS} based on the ideas of ADMM are
\begin{eqnarray}\label{ADMM-PDS}
\begin{cases}
{\bm  w}^{k+1}=\arg\min_{{\bm  w}}
\mathcal{L}_{\beta}({\bm  x}^{k},{\bm  y}^k,{\bm  w};{\bm  z}^k),\\
({\bm  x}^{k+1},{\bm  y}^{k+1})=\arg\min_{{\bm  x},{\bm  y}}
\mathcal{L}_{\beta}({\bm  x},{\bm  y},{\bm  w}^{k+1};{\bm  z}^k),\\
{\bm  z}^{k+1}={\bm  z}^k+\beta({\bm  x}^{k+1}-{\bm  w}^{k+1}).\\
\end{cases}
\end{eqnarray}

Before solving the  ${\bm  w}$-subproblem in \eqref{ADMM-PDS},
we first recall results for a proximal operator.
In \cite[Proposition 7.1]{liu2017further} and \cite[Section 2]{lou2018fast}, proximal operator
\begin{equation}\label{L1L2-ProximalMap}
\arg\min_{\bm{x}}\frac{1}{2}\|{\bm  x}-{\bm  b}\|_2^2+(\mu_1\|{\bm  x}\|_1-\mu_2\|{\bm  x}\|_2), \ \ \ \ \mu_1\geq\mu_2>0
\end{equation}
has an explicit formula for ${\bm  x}$,   denoting  $\text{Prox}_{\mu_1\ell_1-\mu_2\ell_2}({\bm  b})$.
The minimization in \eqref{ADMM-PDS} respecting to ${\bm  w}$ has the following closed-form solution
\begin{equation}\label{w-subproblem}
{\bm  w}^{k+1}=\text{Prox}_{\frac{\rho_1}{\beta}\ell_1-\frac{\alpha\rho_1}{\beta}\ell_2}
\bigg({\bm  x}^k+\frac{{\bm  z}^k}{\beta}\Big)\bigg).
\end{equation}

Next, we turn our attention to the $({\bm x},{\bm y})$-subproblem. The ${\bm  x}$-subproblem has closed-form solution as follows
\begin{equation}\label{x-subproblem}
{\bm  x}^{k+1}=\big({\bm B}^{T}{\bm B}+\beta{\bm I}\big)^{-1}
\bigg({\bm  B}^{T}({\bm  y}^{k}+{\bm c})+\beta\big({\bm  w}^{k+1}-\frac{{\bm z}^{k}}{\beta}\big)\bigg),
\end{equation}
where the inverse of ${\bm B}^{T}{\bm B}+\beta{\bm I}$ is computed by Woodbury matrix identity.
The ${\bm  y}$-subproblem also has closed-form solution as follows
\begin{equation}\label{y-subproblem}
{\bm  y}^{k+1}=\text{Proj}_{\mathcal{B}^{\infty}(\eta)}
\left({\bm  B}{\bm  x}^{k}-{\bm  c}\right),
\end{equation}
where $\text{Proj}_{\mathcal{B}^{\infty}(\eta)}({\bm  x})$ is a projection on the ball $\mathcal{B}^{\infty}(\eta)$, i.e.,
$$
\text{Proj}_{\mathcal{B}^{\infty}(\eta)}({\bm  x})_j=\min( \max(x_j,-\eta),\eta), ~j=1,\ldots,n .
$$

On account of the above discussions, the  effective Algorithm to  approximately solve \eqref{VectorL1-alphaL2-DS-penalty-Equivalent} is
summarized as Algorithm $1$.

\medskip
\noindent\rule[0.25\baselineskip]{\textwidth}{1pt}
\label{al:LADMg}
\centerline {\bf Algorithm $1$: ADMM for solving \eqref{VectorL1-alphaL2-DS-penalty-Equivalent}}\\
{\bf Input :}${\bm  A}$, ${\bm  b}$, $\eta$, $0<\alpha\leq 1$,  $\lambda$,  $\beta$. \\
{\bf Initials}:\  $({\bm  x},{\bm  y},{\bm  w};{\bm  z})
=({\bm  x}^0,{\bm  y}^0,{\bm  w}^0;{\bm  z}^{0})$, $k=0$.\\
{\bf Circulate} Step 1--Step 4 until ``some stopping criterion is satisfied":  

 ~{\bf Step 1:} Compute ${\bm  w}^{k+1}$ by \eqref{w-subproblem}.

~{\bf Step 2:} Compute ${\bm  x}^{k+1},{\bm  y}^{k+1}$ by \eqref{x-subproblem} and \eqref{y-subproblem},respectively.

~{\bf Step 3:} Update dual variables
\begin{equation*}
{\bm  z}^{k+1}={\bm  z}^k+\beta({\bm  x}^{k+1}-{\bm  w}^{k+1})
\end{equation*}

~{\bf Step 4:} Update $k$ to $k+1$.\\
{\bf Output:}  $\bm {x}^{k}$.\\
\noindent\rule[0.25\baselineskip]{\textwidth}{1pt}

\begin{remark}
In Algorithm 1,
$\lambda,\alpha,\eta$ are model parameters satisfying $\lambda>0$, $0<\alpha\leq 1$ and $\eta>0$, and $\beta>0$ is the regularized parameter in ADMM algorithm.
\end{remark}

\section{Numerical Experiments}\label{s5}

In the section, we  present numerical experiments for the recovery of  sparse signals  to demonstrate the performance of $\ell_1-\alpha\ell_2$-DS \eqref{VectorL1-alphaL2-DS}.

In our experiments, our method $\ell_1-\alpha \ell_2$-DS \eqref{VectorL1-alphaL2-DS} is
compared with $\ell_1$-DS \eqref{VectorL1-DS} implemented by  linear ADMM \cite{wang2012linearized},
 and  the $\ell_p$-DS \eqref{VectorLp-DS} as follows
\begin{equation}\label{VectorLp-DS}
\min_{{\bm  x}\in\mathbb{R}^n}~\|{\bm  x}\|_{p}^{p}~~\text{subject~ to}~ \|{\bm  A}^{\top}({\bm  b}-{\bm  A}{\bm  x})\|_\infty\leq\eta
\end{equation}
where $0<p<1$.  Similarly,
An  effective algorithm for solving  \eqref{VectorLp-DS} can be  developed based on  Algorithm 1.
We only need
\begin{equation*}
\bm {w}^{k+1}=\text{Prox}_{\frac{1}{\beta}\ell_p}
\bigg({\bm  x}^k+\frac{{\bm  z}^k}{\beta}\bigg)
\end{equation*}
instead of \eqref{w-subproblem} for  Algorithm 1, where the proximal operator
\begin{equation*}
\text{Prox}_{\mu\ell_p}
\left(\bm {b}\right):=\arg\min_{\bm {x}}\frac{1}{2}\|\bm {x}-\bm {b}\|_2^2+\mu\|\bm {x}\|_p^p,~ \mu>0
\end{equation*}
 has an explicit formula for $\bm {x}$  in \cite{xu2012lregularization,marjanovic2012optimization},   denoting  $\text{Prox}_{\mu \ell_p}(\bm {b})$.

And  we apply the proposed Algorithm 1 for the  $\ell_1-\alpha \ell_2$-DS \eqref{VectorL1-alphaL2-DS} to reconstruct sparse signals in
the cases of  Gaussian , Symmetric $\tilde{\alpha}$-stable ($S\tilde{\alpha}S$)  and uniform noises, which have been defined  in Section \ref{s1.1}.

In our experiments, we test two  measurement matrices defined  in Subsection \ref{s1.1}, which
have different coherence.
Let  ${\bm  x}^{0}\in\mathbb{R}^{n}$ be a simulated $s$-sparse signal,
 where the support of ${\bm  x}^{0}$ is a random index set and
  the $s$ non-zero entries obey the Gaussian distribution $\mathcal{N}(0,1)$.
In addition, the signal ${\bm  x}^{0}$ is normalized to have a unit energy value.
Let ${\bm  \hat{{\bm  x}}}$ be the estimation of ${\bm  x}^{0}$ via  each solver.

Each provided result is an average over 100 independent tests.
All experiments
are performed under Windows Vista Premium and MATLAB v9.1 (R2016b)
running on a Huawei laptop-qolkaflg with an Intel(R) Core(TM)i5-8250U CPU at 1.8 GHz and 8195MB RAM of memory.

We take a self-adapting strategy to update the parameter $\alpha$. The initial value was $\alpha^{0}=0.1$ and then
it was adjusted iteratively by the strategy
\begin{equation}\label{adjustedstrategy}
\alpha^{k+1}=
\begin{cases}
\alpha^{k}, &\text{if~mod}(k,5)\neq 0,\\
\min\{1.5\alpha^{k},1\}, &\text{if~mod}(k,5)=0.
\end{cases}
\end{equation}
Recall that $\alpha\leq 1$ is required to ensure $\|\bm x\|_{\alpha,1-2}\geq 0$ for any $\bm x$.
The above strategy of choosing $\alpha$ clearly satisfies this condition.

\subsection{Observations with Gaussian Noise }\label{s5.1}

First,   the measurement matrix ${\bm  A}$  is   Gaussian matrix. We follow the method in \cite{candes2007dantzig} to generate Gaussian matrix ${\bm  A}$ whose columns all have the unit norm. More
specifically, we first generated an $m\times n$ matrix with independent Gaussian entries and
then normalized each column with the unit norm. After that, we randomly choose a
sample set $S$ with cardinality $s$. Then, the coefficient vector ${\bm  x}^{0}$ was generated by
\begin{equation}
x_i^{0} =
\begin{cases}
\xi_i (1+|c_i|), &  i\in S,\\
0,                 &\text{otherwise},
\end{cases}
\end{equation}
where $\xi_i\in \mathcal{U}(-1, 1)$ (i.e., the uniform distribution on the interval $(-1, 1)$) and $c_i\sim \mathcal{N}(0, 1)$. Finally, the vector of observations ${\bm  b}$ was generated by ${\bm  b} =\bm { Ax}+{\bm  e}$ with ${\bm  e}\sim \mathcal{N}(0, \sigma^2 )$.

To compare with the $\ell_1$-DS in \cite{candes2005magic} and $\ell_p$-DS \eqref{VectorLp-DS}, we test the same cases of $\sigma$, i.e., $\sigma=0.01$ and $\sigma=0.05$, with $(n, p, s)=(72i, 256i, 8i)$ for $i=1,2,3$. Here the coherence $\mu({\bm  A})$ decreases from 0.50 to 0.25.  For each case, we
generated ten different problems and reported the average performance. As in \cite{candes2007dantzig}, the quality of the Dantzig selector is measured by
\begin{equation}\label{quality.Dantzigselector}
\rho_{orign}^2 = \frac{\sum_{j}|\hat{x}_j -x_i^{0}|^2}{\sum_{j}\min\{(x_j^{0})^2,\sigma^2\}}
\ \ {\rm and}\ \
\rho^2 = \frac{\sum_{j}|\tilde{x}_j -x_i^{0}|^2}{\sum_{j}\min\{(x_j^{0})^2,\sigma^2\}},
\end{equation}
where $\hat{{\bm  x}}$ denotes the Dantzig selector via solving \eqref{VectorL1-DS},  \eqref{VectorL1-alphaL2-DS} and \eqref{VectorLp-DS}, and $\tilde{{\bm  x}}$ is the corresponding refined Dantzig
selector after the bias-removing two-stage procedure in \cite{candes2007dantzig}. In \cite{lu2012alternating}, $\rho_{orign}^2$
 and $\rho^2$ are named as the preprocessing and postprocessing errors, respectively. Note that $\rho_{orign}^2$
 and $\rho^2$ are two measurements on the  performance of the Dantzig selector. Obviously, we are
pursuing better selectors which have smaller values of them.

From Table \ref{tab:GaussianPlusGaussian}, we repeat the numerical performance $\ell_1$-DS, $\ell_p$-DS and the proposed  $\ell_1-\alpha \ell_2$-DS
for  $\sigma=0.01,0.05$. We display the average values of $\rho_{orign}^2,\rho^2$,  the number of iterations (Iter), and the computing time in seconds (``Time (s)") over $100$ independent trials.  The data in Table \ref{tab:GaussianPlusGaussian} show the efficiently of the proposed $\ell_1-\alpha\ell_2$-DS. We can see that the $\ell_1$-DS needs the least time, following by the propose method. However, our method has the best performance in terms of the value of $\rho_{orign}^2$ and $\rho^2$.

\begin{table}[!t]
\setlength{\tabcolsep}{4.0pt}\small
\caption{{\rm Numerical results for Gaussian  matrix with unit column norms.}}\label{tab:GaussianPlusGaussian}
\vspace{-2mm}
  \begin{center}
      \begin{tabular}{|l|l|  c c c| c c c |}\hline\specialrule{0em}{1pt}{1pt}
\multirow{1}{*}{i} &\multirow{2}{*}{ Algorithms}  &\multicolumn{3}{c|}{$\sigma=0.01$}&\multicolumn{3}{c|}{{$\sigma=0.05$}}\\
\cmidrule(lr) {3-5}\cmidrule(lr){6-8}
	\quad     &\quad      &Time(s)   &$\rho^2$ &$\rho_{orign}^2$  &Time(s)  &$\rho^2$ &$\rho_{orign}^2$  \\ \hline
            \specialrule{0.0em}{2.0pt}{2.0pt}
	$i=1$ &$\ell_1$-DS  &0.11  &1.24   &17.54   &0.11  &1.12  &20.51\\\hline
           &$\ell_{0.9}$-DS   &1.15  &1.07   &6.57  &1.26  &1.13   &14.87\\\hline
           &$\ell_{0.5}$-DS   &1.16 &1.07   &8.55  &1.24  &1.10  &14.41 \\\hline
           &$\ell_{0.1}$-DS   &1.36  &1.17   &6.53  &1.11  &1.06  &13.00 \\\hline
 &$\ell_1-\alpha\ell_2$-DS   &0.53  &1.03 &5.15    &0.51 &1.02 &9.10\\\hline
			\specialrule{0.0em}{2.0pt}{2.0pt}	
	$i=2$ &$\ell_1$-DS   &0.34  &1.18   &21.65   &0.26  &1.09   &20.91\\\hline
           &$\ell_{0.9}$-DS   &7.12  &1.10   &20.43   &7.09  &1.10   &19.97\\\hline
           &$\ell_{0.5}$-DS   &7.61  &1.05    &16.41   &6.10  &1.04   &18.96\\\hline
           &$\ell_{0.1}$-DS  &7.93  &1.13    &17.24  &7.48  &1.11 &21.19\\\hline
&$\ell_1-\alpha\ell_2$-DS   &2.28  &1.04  &8.32    &4.16  &1.02  &7.53\\\hline
			\specialrule{0.0em}{2.0pt}{2.0pt}	
	$i=3$ &$\ell_1$-DS   &1.02  &1.16   &24.61   &1.52  &1.13   &24.59\\\hline
           &$\ell_{0.9}$-DS   &23.23  &1.15   &22.05   &13.95  &1.12   &21.45\\\hline
           &$\ell_{0.5}$-DS   &23.39  &1.14  &23.05   &13.94  &1.11  &21.56\\\hline
           &$\ell_{0.1}$-DS   &22.54  &1.16  &23.15  &14.04  &1.18   &22.59\\\hline
 &$\ell_1-\alpha\ell_2$-DS   &4.38  &1.09   &8.24   &10.22  &1.06   &8.52 \\\hline
			\specialrule{0.0em}{2.0pt}{2.0pt}	
		\end{tabular}
  \end{center}
\end{table}

Next,  we consider that  the measurement matrix ${\bm  A}$ is  oversampled partial DCT matrix.  We use the average of  the signal-to-noise ratio (SNR) in dB,
\begin{equation}
\text{SNR}(\hat{\bm  x},{\bm  x}_0)=20\log_{10}\frac{\|{\bm  x}_0\|_2}{\|\hat{\bm  x}-{\bm  x}_0\|_2},
\end{equation}
over $100$  independent trials
as our performance measure,  where   $\hat{\bm x}$ is the reconstructed signal. We display SNR of different algorithms to recover sparse signals over 100 repeated trials for $m = 64, n = 256$ and different sparsity $s$.
 The oversampled partial DCT matrix ${\bm A}\in\mathbb{R}^{m\times n}$ with $F = 10$ has high
coherence with $\mu(\bm A)>0.99$.  From the Table \ref{tab:DS-oversampledDCT-Gaussian-sparsity},
 we see that SNR of the proposed $\ell_1-\alpha\ell_2$-DS  are higher than that of  $\ell_1$-DS, $\ell_p$-DS.

\begin{table}[t] 
\setlength{\tabcolsep}{5pt}\small
{\caption{\rm
The average of SNR  over 100  independent trials for different algorithms, the  oversampled DCT ${\bm A}\in\mathbb{R}^{m\times n}$ with $m = 64, n = 256, F = 10$, the   measurements corrupted by Gaussian noises with two different noise levels  $\sigma$  and
 different sparsity.}\label{tab:DS-oversampledDCT-Gaussian-sparsity}}
	\begin{center}
		\begin{tabular}{c| c| c c c  c c c c  c c c c  c c c c}\hline
			$\sigma $ 
& \backslashbox
{Alg.} {Sparsity}            &1     &2    &4      &6   &8     &10   &12 \\ \hline
$10^{-3}$ &$\ell^{1}$-DS &29.52   &26.91  &26.00  &24.01 &18.00   &16.00  &11.42\\
       \quad&$\ell^{0.9}$-DS&38.61  &37.43 &31.04  &25.75  &14.71   &10.50  &8.06 \\
       \quad&$\ell^{0.5}$-DS&38.02  &37.13  &31.41  &24.64 &16.69  &12.24  &7.52 \\
       \quad& $\ell^{0.1}$-DS&37.77 &36.38  &31.63  &20.99   &12.95    &5.77   &3.57 \\
\quad& $\ell_1-\alpha\ell_2$-DS &41.35 &36.74 &34.91 &32.13 &29.13 &26.23 &23.80 \\
\hline  
$10^{-2}$ &$\ell^1$-DS &11.73   &9.78  &9.45 &8.75 &6.08   &4.64 &5.13\\
           \quad&$\ell^{0.9}$-DS   &14.72  &14.61 &12.82  &10.37  &8.47  &8.01 &5.85\\
           \quad&$\ell^{0.5}$-DS   &15.77 &15.28 &14.37  &12.19 &9.42  &7.01 &4.70\\
           \quad& $\ell^{0.1}$-DS &16.47  &16.39   &10.69  &10.09   &4.91   &3.47  &2.75 \\
          \quad& $\ell_1-\alpha\ell_2$-DS &24.68 &17.41 &13.89  &13.52   &11.26 &8.57 &6.12 \\
          \hline
		  \specialrule{0.0em}{2.0pt}{2.0pt}
		\end{tabular}
	\end{center}
\end{table}

\subsection{Observations with Impulsive Noise }\label{s5.2}

In this subsection, we consider that  the observation is corrupted by impulsive noise.
 And  let the measurement matrix  ${\bm A}\in\mathbb{R}^{m\times n}$  with $m=64,n=256$  first be    Gaussian matrix,  which has coherence $0.45<\mu(\bm A)<0.50$.
Next, ${\bm A}\in\mathbb{R}^{m\times n}$  with $m=64,n=256$ is the oversampled DCT matrix with $F=10$, which has high coherence $\mu(\bm A)>0.99$.

Tables \ref{tab:DS-Gaussian-Cauchy-sparsity4} and \ref{tab:DS-oversampledDCT-Cauchy-sparsity4}
present  the average of  SNR over  100  independent trials  for the $\ell_1$-DS, $\ell_p$-DS ($0<p\leq1$)  and the proposed $\ell_1-\alpha\ell_2$-DS versus the sparsity $s$ in the $S\tilde{\alpha} S$ noise with  $\tilde{\alpha}=1$ (Cauchy noise),
$\tilde{\delta}=0$ and $\gamma =10^{-4},10^{-3}$. Table \ref{tab:DS-Gaussian-Cauchy-sparsity4}  shows that the proposed $\ell_1-\alpha\ell_2$-DS  provides the best robust performance no matter
the measurement matrix ${\bm  A}$ has small or high coherence.

\begin{table}[t] 
\setlength{\tabcolsep}{5pt}\small
{\caption{\rm
The average of SNR over 100  independent trials for different algorithms, Gaussian matrix $\bm{A}\in\mathbb{R}^{m\times n}$ with $m=64,n=256$, the measurements corrupted by Cauchy noises with two different  levels  $\gamma$,
  and different sparsity.}\label{tab:DS-Gaussian-Cauchy-sparsity4}}
\begin{center}
		\begin{tabular}{c| c| c c c  c c c c  c c c c  c c c c}\hline
			$\gamma $ 
& \backslashbox
{Alg.} {Sparsity}            &1     &2    &4      &6   &8     &10   &12 \\ \hline
$10^{-4}$ &$\ell_{1}$-DS &35.52  &34.25  &33.86  &33.45   &30.70 &24.31   &24.10   \\
    \quad&$\ell_{0.9}$-DS&39.00   &36.57  &35.16  &33.97 &30.07   &33.98  &26.94\\
       \quad&$\ell_{0.5}$-DS&37.50   &35.42  &35.11  &34.97 &25.69   &22.87  &20.40\\
       \quad&$\ell_{0.1}$-DS&38.93   &40.10  &35.72  &20.32 &18.20   &11.29  &9.98\\
          \quad& $\ell_1-\alpha\ell_2$-DS &46.41 &41.10 &41.01 &40.91 &40.27 &38.91 &28.10 \\
\hline  
$10^{-3}$ &$\ell_1$-DS &10.59 &12.04 &10.01 &6.85 &5.94 &5.48  &4.07          \\
           \quad&$\ell_{0.9}$-DS  &18.83 &16.90 &12.63 &10.15 &9.66 &8.95  &8.76          \\
           \quad&$\ell_{0.5}$-DS  &17.00  &15.44  &14.63 &11.22  &10.78  &10.54 &9.27 \\
           \quad&$\ell_{0.1}$-DS  &19.13 &17.15 &11.17 &11.05 &9.24 &7.19  &6.02          \\
          \quad& $\ell_1-\alpha\ell_2$-DS &24.26 &18.59 &17.84 &17.25 &16.90 &14.49 &13.21    \\
          \hline
		  \specialrule{0.0em}{2.0pt}{2.0pt}
		\end{tabular}
	\end{center}
\end{table}

\begin{table}[t] 
\setlength{\tabcolsep}{5pt}\small
{\caption{\rm
The average of  SNR over 100  independent trials for   different algorithms, the oversampled DCT matrix $\bm{A}\in\mathbb{R}^{m\times n}$ with $m=64,n=256, F=10$, the measurements corrupted by two different  Cauchy noise levels  $\gamma$ and
 different sparsity.}\label{tab:DS-oversampledDCT-Cauchy-sparsity4}}
	\begin{center}
		\begin{tabular}{c| c| c c c  c c c c  c c c c  c c c c}\hline
			$\gamma $ 
& \backslashbox
{Alg.} {Sparisty}            &1     &2    &4      &6   &8     &10   &12 \\ \hline
$10^{-4}$ &$\ell_{1}$-DS &31.90  &30.25 &29.19   &27.55  &26.48 &21.94   &16.85\\
    \quad&$\ell_{0.9}$-DS&31.43  &29.95  &28.41  &26.08 &20.73   &17.85  &13.60 \\
       \quad&$\ell_{0.5}$-DS&35.76  &29.28  &25.56  &23.48 &16.71   &14.14  &7.86 \\
    \quad&$\ell_{0.1}$-DS&34.66  &33.05  &31.05  &26.07 &20.49   &16.00  &12.41 \\
          \quad& $\ell_1-\alpha\ell_2$-DS &35.27 &33.51 &31.82 &31.63 &28.06 &27.32 &22.96 \\
\hline  
$10^{-3}$ &$\ell_1$-DS &14.42   &12.23  &8.18 &6.01 &5.52 &3.75   &2.99\\
         \quad&$\ell_{0.9}$-DS  &15.86 &13.72  &12.34 &10.12   &9.47 &7.19  &4.97\\
           \quad&$\ell_{0.5}$-DS  &16.35 &14.78  &13.84 &12.56      &8.78  &3.17  &2.75\\
           \quad&$\ell_{0.1}$-DS  &18.97 &14.11  &10.83 &8.72      &6.18  &4.41  &3.26\\
          \quad& $\ell_1-\alpha\ell_2$-DS &19.32 &13.61 &10.06  &8.29  &5.36 &5.35  &3.44 \\
          \hline
		  \specialrule{0.0em}{2.0pt}{2.0pt}
		\end{tabular}
	\end{center}
\end{table}

\newpage

\subsection{Observations with Uniform Noise }\label{s5.3}

In this subsection, we consider the  observation is corrupted by uniform noise. And the
the measurement matrix  $\bm{A}$ is same with  the presented  Gaussian and  oversampled DCT in Subsection \ref{s5.2}.
We take noise levels $\varsigma=10^{-3},10^{-2}$ and display the average of SNR over  100  independent trials in Tables \ref{tab:DS-Gaussian-Uniformlydistribution-sparsity} and \ref{tab:DS-OverDCT-Uniformlydistribution-sparsity}.
 These results show that the Dantzig selector also work efficiently for uniform noise, which similar as that of $\ell_{\infty}$ constraint. It also implies that our proposed method is better than $\ell_1$-DS and $\ell_p$-DS for both Gaussian and oversamped DCT  matrices.

\begin{table}[t] 
\setlength{\tabcolsep}{5pt}\small
	{\caption{\rm  The average  of SNR  over  100  independent trials for  different algorithms,  Gaussian matrix $\bm{A}\in\mathbb{R}^{m\times n}$ with $m=64,n=256$, the measurements corrupted by  two different  uniform noise levels  $\varsigma$ and
 different sparsity. }\label{tab:DS-Gaussian-Uniformlydistribution-sparsity} }
\begin{center}
		\begin{tabular}{c| c| c c c  c c c c  c c c c  c c c c}\hline
			$\varsigma$ 
& \backslashbox
{Alg.} {$m/n$}            &1     &2    &4      &6   &8     &10   &12 \\ \hline
$10^{-3}$ &$\ell_{1}$-DS &42.64  &39.97 &39.17  &38.50 &38.29  &35.29   &31.34 \\
      \quad&$\ell_{0.9}$-DS&44.53 &43.91 &43.67  &43.53 &43.12 &42.61  &41.58\\
       \quad&$\ell_{0.5}$-DS&48.06 &47.79   &46.55  &46.33 &45.89   &40.87 &29.98\\
       \quad&$\ell_{0.1}$-DS&45.34 &44.62 &44.08 &43.21   &43.36 &42.50   &42.18\\
\quad& $\ell_1-\alpha\ell_2$-DS &61.04 &57.06 &56.91 &56.12 &55.50   &55.33 &55.02 \\
\hline  
$10^{-2}$ &$\ell_1$-DS  &24.35 &24.18 &23.76 &23.54  &23.09   &21.85  &17.75\\
           \quad&$\ell_{0.9}$-DS   &24.49 &23.93 &23.66 &22.86 &21.25 &20.04  &18.79\\
             \quad&$\ell_{0.5}$-DS &25.38  &24.00 &22.10 &19.43 &17.71 &16.61  &16.55\\
              \quad&$\ell_{0.1}$-DS   &26.65  &23.09 &21.62 &20.79 &20.73 &19.92  &17.32\\
\quad& $\ell_1-\alpha\ell_2$-DS &29.01 &25.61  &25.11 &24.55  &23.69  &22.17 &21.75 \\
          \hline
		  \specialrule{0.0em}{2.0pt}{2.0pt}
		\end{tabular}
	\end{center}
\end{table}

\begin{table}[t] 
\setlength{\tabcolsep}{5pt}\small
	{\caption{\rm The average of SNR over  100  independent trials for  different algorithms, the oversampled DCT matrix $\bm{A}\in\mathbb{R}^{m\times n}$ with $m=64,n=256, F=10$, the measurements corrupted by  two different  uniform noise levels  $\varsigma$
and  different sparsity. }\label{tab:DS-OverDCT-Uniformlydistribution-sparsity} }
\begin{center}
		\begin{tabular}{c| c| c c c  c c c c  c c c c  c c c c}\hline
			$\varsigma$ 
& \backslashbox
{Alg.} {$m/n$}            &1     &2    &4      &6   &8     &10   &12 \\ \hline
$10^{-3}$ &$\ell-{1}$-DS &33.67 &31.92 &31.00 &27.81 &22.82  &18.02 &16.94 \\
     \quad&$\ell_{0.9}$-DS&37.24 &36.99   &36.53  &30.95 &30.90 &27.03   &21.09\\
       \quad&$\ell_{0.5}$-DS&38.71   &37.08  &34.93 &33.32  &32.42 &24.17   &16.82\\
        \quad&$\ell_{0.1}$-DS&38.82   &37.71  &36.17 &34.21  &25.42 &23.51   &19.74\\
\quad& $\ell_1-\alpha\ell_2$-DS &44.05 &39.18 &37.66 &36.10 &32.72 &26.38 &27.15 \\
\hline  
$10^{-2}$ &$\ell-1$-DS &16.02   &15.71  &13.52 &9.77 &9.56   &7.64 &7.06\\
             \quad&$\ell_{0.9}$-DS   &19.65 &19.14 &18.37 &16.87  &16.02 &14.28  &12.07\\
             \quad&$\ell_{0.5}$-DS  &21.49  &21.61 &18.90  &16.66  &16.06  &14.78 &11.12\\
             \quad&$\ell_{0.1}$-DS   &22.44  &22.81 &18.79  &17.45  &14.66  &12.77  &7.17\\
\quad& $\ell_1-\alpha\ell_2$-DS &26.81 &24.49 &24.22  &22.02 &21.14  &19.33  &17.31 \\
          \hline
		  \specialrule{0.0em}{2.0pt}{2.0pt}
		\end{tabular}
	\end{center}
\end{table}

\section{Conclusions }\label{s6}
In this paper, we consider the signal reconstruction under Dantzig selector constraint via
$\ell_{1}-\alpha\ell_{2}~(0<\alpha\leq 1)$ minimization. First, we introduce  the    $\ell_{1}-\alpha\ell_{2}$-DS \eqref{VectorL1-alphaL2-DS} to recover  signals ${\bm  x}$ from ${\bm  b}={\bm  A}{\bm  x}+{\bm  e}$.
Next,  we show a sufficient  condition based on $(\ell_2,\ell_1)$-RIP to guarantee the stable recovery of  signal ${\bm  x}$ from ${\bm  b}={\bm  A}{\bm  x}+{\bm  e}$ via \eqref{VectorL1-alphaL2-DS} (see Theorem \ref{StableRecoveryviaVectorL1-alphaL2-DS}).
And  Based on the high order  classical RIP (i.e., $(\ell_2,\ell_2)$-RIP ),
  we develop a sufficient  condition   for  the stable reconstruction of signals ${\bm  x}$ via \eqref{VectorL1-alphaL2-DS}  (see Theorem \ref{New-StableRecoveryviaVectorL1-alphaL2-DS}
  applying the technique of the convex combination for $\ell_{1}-\ell_{2}$. Last, we show an effective algorithm based on ADMM  to  solve the $\ell_1-\alpha \ell_2$-DS  \eqref{VectorL1-alphaL2-DS}
  Furthermore, we present numerical experiments for the  sparse signal reconstruction
 in the cases of  Gaussian,  impulsive  and uniform noises. Results  demonstrate the efficiency of $\ell_{1}-\alpha\ell_{2}$-DS, which is different from that of $\ell_2$, $\ell_1$ and $\ell_{\infty}$ data fitting terms.
  What should point out is that, this is the first paper which explores the performances of Dantzig selector for different type noises. The  proposed method also outperform than existing methods no matter the measurement matrix has high or small coherence.

\section*{Acknowledgments}
The  project is partially supported by the Natural Science Foundation of China (Nos. 11871109, 11901037, 72071018), the NSAF (Grant No. U1830107) and the Science  Challenge Project (TZ2018001).
The authors thanks Professors Wengu Chen and Qiyu Sun for their help in the preparation of this paper.


\hskip\parindent

\bibliographystyle{plain}
\bibliography{New_L1_L2_DS_reference}

\end{document}